\documentclass[11pt]{amsart}
\usepackage{amssymb}
\usepackage{amsthm}
\usepackage{amsfonts}
\usepackage{amsmath}
\usepackage{amscd}
\usepackage{enumitem}
\usepackage[utf8]{inputenc}
\usepackage[margin=1in]{geometry}
\usepackage{graphicx}
\usepackage{tikz}
\usepackage{mathrsfs}
\usepackage{hyperref}
\usepackage{comment}

\newtheorem{theorem}{Theorem}[section]
\newtheorem{lemma}[theorem]{Lemma}
\newtheorem{definition}[theorem]{Definition}
\newtheorem{proposition}[theorem]{Proposition}

\newtheorem{remark}[theorem]{Remark}

\newcommand{\C}{\mathbb{C}}
\newcommand{\R}{\mathbb{R}}
\newcommand{\Z}{\mathbb{Z}}

\renewcommand{\AA}{\mathcal A}
\renewcommand{\L}{\mathfrak L}
\newcommand{\al}{\alpha}
\newcommand{\be}{\beta}
\newcommand{\ga}{\gamma}
\newcommand{\de}{\delta}
\newcommand{\epsi}{\varepsilon}
\newcommand{\F}{\mathcal{F}}
\newcommand{\sgn}{\operatorname{sgn}}
\newcommand{\rr}{\mathbf r}
\newcommand{\bfc}{\mathfrak{c}}
\newcommand{\bfd}{\mathfrak{d}}
\newcommand{\bfz}{\mathfrak{z}}
\newcommand{\thq}{\theta_q}
\newcommand{\z}{\zeta}
\newcommand{\zm}{\zeta_-}
\newcommand{\zp}{\zeta_+}
\newcommand{\D}{\mathcal{D}}
\newcommand{\Y}{\mathfrak{Y}}
\newcommand{\K}{\mathfrak{K}}
\newcommand{\tilP}{\widetilde{P}}
\newcommand{\tilQ}{\widetilde{Q}}
\newcommand{\hatK}{\widehat{\mathfrak{K}}}
\newcommand{\hatk}{\widehat{K}}
\newcommand{\tilK}{\widetilde{K}}
\newcommand{\X}{\mathfrak{X}}
\newcommand{\Om}{\Omega}
\newcommand{\om}{\omega}

\newcommand{\Res}{\text{Res}}
\newcommand{\Conf}{\operatorname{Conf}}

\numberwithin{equation}{section}

\begin{document}

\title{The elliptic tail kernel}

\author{Cesar Cuenca, Vadim Gorin and Grigori Olshanski}


\maketitle

\begin{abstract}
We introduce and study a new family of $q$-translation-invariant determinantal point processes on the two-sided $q$-lattice.
We prove that these processes are limits of the $q$--$zw$ measures, which arise in the $q$-deformation of harmonic analysis on $U(\infty)$, and express their correlation kernels in terms of Jacobi theta functions.
As an application, we show that the $q$--$zw$ measures are diffuse. Our results also hint at a link between the two-sided $q$-lattice and rows/columns of Young diagrams.
\end{abstract}

\setcounter{tocdepth}{1}
\tableofcontents

\section{Introduction}

\subsection{Preface}

The subject of this paper is the study of a family of ``tail processes'' associated to certain random point processes of representation-theoretic origin.
All the point processes in this article are determinantal and thus our study will focus on their correlation kernels.

Recall that a random point process on a locally compact space $\X$ is defined by a probability measure on the space $\operatorname{Conf}(\X)$ of locally-finite point configurations.
Informally, its associated tail process describes the random point configuration near infinity.
A simple example is provided by the well-known family of \emph{Poisson--Dirichlet distributions} $PD(\tau)$, $\tau > 0$, which describe asymptotics of random ranked relative frequencies in various contexts, see \cite{Ki}, \cite{Pit} and references therein.
In this case, the space is $\X = (0, 1]$.
Locally-finite point configurations accumulate near $0$, which is excluded from $\X$ and plays the role of the boundary point at infinity.
Under the change of variables $x = e^{-t}$, the interval $(0, 1]$ is transformed into $[0, +\infty)$, so that the boundary point is transferred to $+\infty$.
Here the tail processes turn out to be stationary Poisson processes: they arise in the limit transition when we set $t = A + s$ and take the limit $A \to +\infty$.

Another example comes from the problem of harmonic analysis on the infinite symmetric group, \cite{BO_1998}, \cite{BO_2000}. That problem leads to a two-parameter family of measures on $\operatorname{Conf}(\R\setminus\{0\})$, called the \emph{$z$--measures} (the two parameters are usually labeled $z, z'$, hence the terminology).
Again, we are interested in the behavior of random point configurations near the origin.
Notice that there is an important difference with the previous example, namely that points on $\X$ can approximate the point $0$ from the right and from the left. As explained in \cite{BO_1998}, see also \cite{O}, the tail process of a $z$--measure lives on the space $\operatorname{Conf}(\R \sqcup \R)$ -- where $\R \sqcup \R$ is the disjoint union of two copies of the real line -- and is stationary in the sense that it is invariant with respect to simultaneous shifts on both of the lines.

Both the $z$--measures and their tail processes belong to the class of \emph{determinantal measures}. A detailed discussion of this notion can be found in \cite{So}; here we only point out that a determinantal measure on $\operatorname{Conf}(\X)$ is uniquely determined by a complex-valued function on $\X\times\X$ called a \emph{correlation kernel}.

In the case of the tail process of a $z$--measure, the correlation kernel can be treated as a certain stationary $2\times 2$ matrix kernel on $\R$.
It is expressed through trigonometric functions and we call it the \emph{matrix trigonometric kernel}.

The matrix trigonometric kernel has some resemblance with the famous sine kernel from random matrix theory. Like the sine kernel, the matrix trigonometric kernel possesses a universality property: it serves as the tail kernel for not only the $z$--measures, but also for other determinantal measures of representation-theoretic origin, see \cite{BO_2005}.
Those other measures arise in the context of harmonic analysis on infinite-dimensional classical Lie groups: unitary, orthogonal and symplectic, see \cite{BO_2005_2} for the unitary picture and \cite{C} for the orthogonal and symplectic pictures.

As shown in \cite{GO}, the problem of harmonic analysis on the infinite-dimensional unitary group admits a kind of quantization, which leads to a new family of determinantal point processes. These processes live on a two-sided $q$-lattice --- a countable subset of $\R\setminus\{0\}$ accumulating near $0$ from both sides (see below). We use the name \emph{$q$--$zw$ measures} for these objects, as the analogous measures in the $q = 1$  case are called the \emph{$zw$ measures} (they depend on four parameters usually labeled $z, z', w, w'$).
The study of the tail processes of the $q$--$zw$ measures is the central topic of this text.
In this direction, we discover a two-parameter family of $q$-translation-invariant kernels on the two-sided $q$-lattice. They are expressed in terms of theta functions. We believe that these kernels are fundamental objects in the realm of determinantal processes on the two-sided $q$-lattice.

The question of quantizations of constructions in asymptotic representation theory has led to an extensive theory in recent years. In particular, \cite{Go} investigated the $q$-versions of the extreme characters of the infinite-dimensional unitary group, \cite{CG} produced such results for the infinite-dimensional orthogonal and symplectic groups, \cite{C_Mac} added Macdonald's parameter $t$ to the theory, \cite{BG1} studied the related Markov processes, and \cite{Sato1}, \cite{Sato2} linked the extreme $q$-characters to representations of inductive limits of compact quantum groups.

There have also been a few combinatorial developments related to the $q$--$zw$ measures: in \cite{O2}, the associated Markov processes and splines were investigated; in \cite{O3}, \cite{CO}, the $q$--$zw$ measures motivated the construction of new families of orthogonal symmetric functions.
In spite of these efforts, and in contrast to the $q=1$ case \cite{BO_2005_2}, the representation--theoretic meaning of the $q$--$zw$ measures is not fully understood at this moment. In particular, we do not know the proper role of the two-sided $q$-lattice in the representation--theoretic picture.

The study of the tail processes of $q$--$zw$ measures is, thus, also motivated by the attempt to better understand the $q$--$zw$ measures themselves.
We partially accomplish that goal --- along the path we show that the $q$--$zw$ measures are diffuse.
Furthermore, while investigating a $q\to 1$ limit of the tail processes, we manage to link the mysterious two-sided $q$-lattice to separate encodings of rows and columns of Young diagrams.

\subsection{The measures $M^{\al, \be, \ga, \de}$}

The two-sided $q$-lattice mentioned above has the form
$$
\L := \{ \ldots, \zm q^{-1}, \zm, \zm q, \ldots \} \sqcup \{ \ldots, \zp q, \zp, \zp q^{-1}, \ldots \},
$$
where $\zp > 0 > \zm$ and $q\in (0, 1)$ are fixed parameters.
The determinantal processes investigated in \cite{GO} are given by certain probability measures $M^{\al, \be, \ga, \de}$ on the space of point configurations $\operatorname{Conf}(\L)$. Here $(\al, \be, \ga, \de)$ is a quadruple of complex numbers subject to some constraints. We continue to use the name $q$--$zw$ measure for $M^{\al, \be, \ga, \de}$, although the origin of the name might not be transparent in our present notation.
The measure $M^{\al, \be, \ga, \de}$ is defined as the $N\to\infty$ limit of the measures
 \begin{equation} \label{eq_measure_def}
  M_N^{\al, \be, \ga, \de}(X)= \frac{1}{Z_N} \prod_{i=1}^N \left(|x_i|\frac{(\al x_i; q)_{\infty}(\be x_i; q)_{\infty}}{(\ga q^{1-N}x_i; q)_{\infty}(\de q^{1-N}x_i; q)_{\infty}}\right) \prod_{1\le i<j\le N}
  (x_i - x_j)^2,
 \end{equation}
on $N$--particle configurations $X=\{x_1,\dots,x_N\}$ on $\L$. Here $Z_N$ is an explicit normalization constant and $(u;q)_\infty=\prod_{i=0}^{\infty}(1-uq^i)$ is the infinite $q$-Pochhammer symbol.

The $q$--$zw$ measure $M^{\alpha, \beta, \gamma, \delta}$ possesses a reflection symmetry property: the transposition of the positive and negative parts of the two-sided $q$-lattice (together with the changes $\zp \to -\zm$ and $\zm \to -\zp$) is equivalent to changing the signs of the parameters $\al, \be, \ga, \de$.
We shall later see that a similar property holds for the tail process, but that the symmetry is destroyed in the limit transition $q \to 1$.

 There exists a function $K^{\al, \be, \ga, \de}(x,y)$ on $\L\times\L$ such that for any $n = 1, 2, \dots$, and any given $n$-element set $\{x_1,\dots,x_n\}$ of $\L$, one has
$$
\text{Prob}\left( \{x_1, \ldots, x_n\} \subset \text{ $M^{\al, \be, \ga, \de}$-- random subset of }\L \right) = \det_{1 \leq i, j \leq n}\left[ K^{\al, \be, \ga, \de}(x_i, x_j) \right].
$$
This means that $M^{\al, \be, \ga, \de}$ is a determinantal measure and  $K^{\al, \be, \ga, \de}(x,y)$ serves as its correlation kernel.
An explicit expression for the kernel $K^{\al, \be, \ga, \de}(x,y)$ was obtained in \cite{GO}; we reproduce it below, see \eqref{basickernel}. The formula involves $q$-factorials, theta functions and the basic hypergeometric function $_2\phi_1$ (a $q$-version of Gauss' hypergeometric function). We call $K^{\al, \be, \ga, \de}(x,y)$ the \emph{basic hypergeometric kernel}.

The reader is referred to \cite{GO} for more detailed information about the measures  $M^{\al, \be, \ga, \de}$ and their connection to the problem of harmonic analysis on the infinite-dimensional unitary group.
Several other kernels of representation-theoretic origin are known in the literature, see \cite{BO_1998}, \cite{BO_2000}, \cite{BO_2005}, \cite{BO_2005_2}, \cite{C}. They involve various hypergeometric functions (up to ${}_4F_3$), but not $q$-hypergeometric ones.

\subsection{Summary of results}

\subsubsection{The tail process}

Note that the  $q$-lattice $\L$ is invariant under the homotheties $T^{\pm1}: x\mapsto q^{\pm1}x$. Given a probability measure $M$ on $\Conf(\L)$, define the transformed measure $T^{-k}M$ by
$$
(T^{-k}M)(A):= M(q^kA), \qquad A\subset \Conf(\L), \quad k\in\Z.
$$

The next result describes the tail processes of the $q$--$zw$ measures, namely the behavior of a random point configuration of $\L$ (distributed according to the $q$--$zw$ measure $M^{\al, \be, \ga, \de}$) near the origin.

\begin{theorem}\label{thm1}[see Theorem \ref{thm:mainlimit} below]
As $k\to +\infty$, the measures $T^{-k}M^{\al,\be,\ga,\de}$ weakly converge to a probability measure $M^{\ga,\de}$, which depends on $\ga$ and $\de$ only.
The measure $M^{\ga, \de}$ is determinantal and a correlation kernel $K^{\ga,\de}(x,y)$ for it is given in \eqref{thetakernel}.
\end{theorem}

The kernel $K^{\ga, \de}(x, y)$ is expressed in terms of theta functions, and we call it the  \emph{elliptic tail kernel}.
By its very definition, the measure $M^{\ga,\de}$ is stationary, i.e., it is invariant under the transformations $T^{\pm 1}$.
After a simple gauge transformation, which does not affect the measure $M^{\ga, \de}$, the kernel $K^{\ga,\de}$ becomes stationary too --- see Proposition \ref{Kperiodic}.

The measures $M^{\ga,\de}$ are responsible for the limit behavior of the random configuration near the origin.
The proof of Theorem \ref{thm1} relies on the analysis of the asymptotics of the basic hypergeometric kernel $K^{\al,\be,\ga,\de}(x,y)$ at $(0, 0)\in\L^2$.

The reader might ask why we associate our tail processes exclusively with the asymptotics at $0$ and do not examine the limit behavior of random configurations at infinity; the reason is that the configurations are almost surely bounded away from $\pm\infty$, so that $0$ is the only accumulation point (see \cite{GO}).

\subsubsection{Absence of atoms} Recall that a measure is said to be \emph{diffuse} if it has no atoms.

\begin{theorem}\label{thm2}[see Theorem \ref{thm:diffuse} below]
The measures $M^{\al,\be,\ga,\de}$ are diffuse.
\end{theorem}

This result is deduced from the existence of the tail process and a simple general criterion to determine whether a determinantal measure is diffuse, which is of some independent interest.

\subsubsection{The projection property.} A kernel on a discrete space $\X$ is said to be a \emph{projection kernel} if it corresponds to a projection operator on the Hilbert space $\ell^2(\X)$ (that is, the operator of orthogonal projection onto a subspace). With a suitable modification, the definition can be extended to nondiscrete spaces too.
In many concrete examples of determinantal processes, the correlation kernels turn out to be projection kernels. The projection property seems to be extremely important: many strong results about determinantal processes rely on it, see e.g. \cite{Bu}, \cite{Lyons}.

\begin{theorem}\label{thm3}[see Theorem \ref{thm:projection} below]
The elliptic tail kernel $K^{\ga, \de}(x,y)$, $x, y\in\L$, is a projection kernel.
\end{theorem}

One may ask whether the basic hypergeometric kernel $K^{\al, \be, \ga, \de}(x, y)$ is a projection kernel too.
We believe that this is true, based on computer experiments --- one plausible approach to prove it is to show that certain orthogonal elements of $\ell^2(\L)$ form a basis of the range of the projection and that $K^{\al, \be, \ga, \de}$ can be expressed in terms of them, see \cite[Rem. 5.8]{GO}.
However, even if we knew that $K^{\al, \be, \ga, \de}(x, y)$ is a projection kernel, Theorem \ref{thm1} does not imply that $K^{\ga, \de}(x, y)$ is also a projection kernel, because the projection property is not necessarily preserved under weak limits.
Thus we had to find a different approach.

The plan of our proof of Theorem \ref{thm3} is as follows.
Using a natural identification of $\L$ with $\Z \sqcup \Z$ one can treat $K^{\ga,\de}(x,y)$ as a kernel on $\Z$ with values in $2\times2$ matrices.
After a simple gauge transformation, that kernel becomes translation--invariant, so the Fourier transform of the gauge-transformed operator corresponds to an operator on the Hilbert space $L^2(\mathbb T;\C^2)$ ($\C^2$-valued functions on the unit circle) given by multiplication by a $2\times2$ matrix-valued function. Then we check that the values of this function are projection matrices. The idea is simple, but its realization required a lot of laborious computations with elliptic functions.

\subsubsection{Degeneration to the matrix trigonometric kernel}
The \emph{particle/hole involution} on $\Z$ is the involutive map $\Conf(\Z)\to \Conf(\Z)$ defined as $X\mapsto \Z\setminus X$ (cf. \cite[Appendix, \S A.3]{BOO}).
Given a measure $M$ on $\Conf(\Z\sqcup\Z)$, let $\widehat M$ stand for the pushforward of $M$ under the particle/hole involution on the first copy of $\Z$. Let us call it the \emph{partial} particle/hole involution. As above, identify the two-sided $q$-lattice $\L$ with $\Z\sqcup\Z$ (the positive part of $\L$ is identified with the first copy of $\Z$ and the negative part of $\L$ with the second copy), which makes it possible to apply the partial particle/hole involution to the measure $M^{\ga,\de}$. Let us denote the resulting measure by $\widehat{M}^{\ga, \de}$.

\begin{theorem}\label{thm:continuous}[see Theorem \ref{thmlimitII} below]
Rescale the lattice $\Z$ by a factor of $\ln{(1/q)}$, so that in the limit $q \to 1^-$ it becomes the real line $\R$. In the limit regime
$$
\z_- = -q^{\bfz_-},\quad \z_+ = q^{\bfz_+},\quad \ga = q^{\bfc - \bfz_+},\quad \de = q^{\bfd - \bfz_+},\quad q \to 1^-,
$$
where $\bfz_-, \bfz_+$ (resp. $\bfc, \bfd$) are fixed real (resp. complex) parameters satisfying certain constraints, the point process given by the measure $\widehat{M}^{\ga, \de}$ weakly converges to the tail process of the $z$--measure with appropriately chosen parameters.
\end{theorem}

Note that the limit transition in question does not exist for the measures $M^{\ga,\de}$: the reason is that the density of particles on the first copy of $\Z$ tends to $1$.

As pointed out in the preface, the tail process of a $z$--measure has the matrix trigonometric kernel as a correlation kernel. Note that the matrix trigonometric kernel is not a projection kernel and is not even symmetric.
However, it possesses a different symmetry property, called \emph{$J$-symmetry}, see \cite{BO_1998}, \cite{BO_2000}.  The origin of the $J$-symmetry property is just the partial particle/hole involution (see \cite{O}) --- this explains why we have replaced $M^{\ga,\de}$ by $\widehat M^{\ga,\de}$. In the context of $z$--measures, the involution is a natural operation, whose origin is the parameterization of Young diagrams (i.e.\ labels of irreducible representations of symmetric groups) via their rows and columns. The particle/hole involution then becomes the transposition of the diagram interchanging rows and columns.

Simultaneously, the positive and negative semi-axes for the $z$--measures and two copies of $\mathbb R$ for their tail processes correspond precisely to the rows and columns of the Young diagrams: as explained in \cite{BO_1998}, the particles on the positive semi-axis encode rows and the particles on the negative semi-axis encode columns. Together with Theorem \ref{thm:continuous} this creates a connection between  
the positive and negative parts of the two-sided $q$-lattice on one side and rows and columns of the Young diagrams on the other side. Hence, one might expect a relation between the two-sided $q$-lattice and some separate quantizations for the rows and columns of Young diagrams, parameterizing irreducible representations. We hope to further develop this point of view in future publications.

\subsubsection{Degeneration to the discrete sine process}
The \emph{discrete sine process} is a stationary determinantal process on $\Z$, depending on a parameter $\phi\in(0,\pi)$. It is given by the correlation kernel
$$
K_{\textrm{sine}}^\phi(m, n) := \begin{cases}
	\displaystyle \frac{\sin(\phi(m - n))}{\pi(m - n)}, &\text{ if }m \neq n,\\
	\displaystyle \frac{\phi}{\pi}, &\text{ if }m = n.
\end{cases}
$$
Note that $K_{\textrm{sine}}^\phi(m, n)$ is a projection kernel. The discrete sine process first appeared as a limit of the Plancherel measure on partitions in \cite{BOO}. It also possesses some universality property, see \cite{BKMM}, \cite{Ok}.

\begin{theorem}\label{thm:discrete}[see Theorem \ref{thmlimitI} below]
Let $\varphi\in (0, \pi)$ and $a\in\R\setminus\{0\}$ be arbitrary parameters.
As $q \to 1^-$, the two-sided $q$-lattice $\L$ fills the entire real line.
Zooming in near $a$, the point process $M^{\ga, \de}$ converges weakly to a discrete sine process in the regime
$$\frac{\ln{\de} - \ln{\ga}}{2i} = \varphi,\quad q \to 1^-.$$
Moreover, the parameter of such discrete sine process depends only on the sign of $a$: if $a > 0$, the parameter is $\pi - \varphi$ and if $a < 0$, the parameter is $\varphi$.
\end{theorem}

\medskip

Finally, let us mention that in all the proofs of this article, we ignore the case $\de = \ga$.
The reason is that this case requires special attention, but the proofs go through, either by employing the formulas at the end of Section \ref{app:elliptic} or by analytic continuation.

\subsection{Organization of the article}

In Sections \ref{sec:qzw} and \ref{sec:ell}, we introduce the basic hypergeometric kernel and the elliptic tail kernel, respectively.
In Section \ref{sec:limit} we prove Theorem \ref{thm1} on the limit from the former kernel to the latter one.
In Sections \ref{sec:diffusity} and \ref{sec:projection}, we prove Theorems \ref{thm2} and \ref{thm3}, respectively.
The \emph{continuous} limit of the elliptic tail kernel, as in Theorem \ref{thm:continuous}, is discussed in Section \ref{sec:continuous}.
The \emph{discrete} limit of Theorem \ref{thm:discrete} occupies Section \ref{sec:discrete}.
Finally, in Appendix \ref{app:jacobi}, we recall Jacobi's imaginary transformation and apply it to deduce some technical estimates.

\subsection*{Acknowledgments}
C.C. and V. G. were partially supported by NSF Grant DMS-1664619.
V. G. was also supported by the NEC Corporation Fund for Research in Computers and Communications and NSF Grant DMS-1855458.

\section{$q$--$zw$ measures and the basic hypergeometric kernel}\label{sec:qzw}

We review some definitions and results from \cite{GO}.

Fix parameters $q \in (0, 1)$ and $\zp > 0 > \zm$ for the rest of the paper.
Define the \emph{two-sided $q$-lattice}
$$\L := \zm q^{\Z} \sqcup \zp q^{\Z} = \{ \ldots, \zm q^{-1}, \zm, \zm q, \ldots \} \sqcup \{ \ldots, \zp q, \zp, \zp q^{-1}, \ldots \}.$$

\begin{definition}\label{admissiblepair}
We say that $(x, y) \in \C^2$ is an \emph{admissible pair} if one of the following holds:
\begin{itemize}
	\item $y = \overline{x} \in \C\setminus \R$ (principal series), or
	\item $\zm^{-1}q^{m} < x, y < \zm^{-1}q^{m+1}$ for some $m\in\Z$, or $\zp^{-1}q^{n+1} < x, y < \zp^{-1}q^n$ for some $n\in\Z$ (complementary series).
\end{itemize}
Note that $xy\in\R$ whenever $(x, y)$ is an admissible pair.

We also say that $(\al, \be, \ga, \de)\in\C^4$ is an \emph{admissible quadruple} if both $(\al, \be)$ and $(\ga, \de)$ are admissible pairs, and additionally $\al\beta < q^2\ga\de$.
\end{definition}

The $q$--$zw$ measures $M^{\al, \be, \ga, \de}$ are probability measures on the space of point configurations $\Conf(\L)$ that depend on admissible quadruples of parameters $(\al, \be, \ga, \de)$.
Equivalently, they are point processes on $\L$ (\cite{DVJ}).
It is known that $M^{\al, \be, \ga, \de}$ (just like any point process on a discrete space) is determined by the sequence $\{ \rho^{\al, \be, \ga, \de}_n(x_1, \ldots, x_n) \}_{n \geq 1}$ of correlation functions:
$$\rho^{\al, \be, \ga, \de}_n(x_1, \ldots, x_n) := \text{Prob}\left( \{x_1, \ldots, x_n\} \subset \text{ $M^{\al, \be, \ga, \de}$-- random subset of }\L  \right), \ x_1, \dots, x_n  \in \L.$$

In \cite[Thm. 5.2]{GO}, for each admissible quadruple $(\al, \be, \ga, \de)$, the $q$--$zw$ measure $M^{\al, \be, \ga, \de}$ is defined as the unique point process on $\L$ whose correlation functions are given by\footnote{In \cite{GO}, there was a sign mistake in the correlation kernel --- the version here is correct.}
\begin{equation*}
\rho^{\al, \be, \ga, \de}_n(x_1, \ldots, x_n) = \det_{1 \leq i, j \leq n}\left[ K^{\al, \be, \ga, \de}(x_i, x_j) \right], \quad n \geq 1,
\end{equation*}
where $x_1, \ldots, x_n \in \L$ are pairwise disjoint, and
\begin{equation}\label{basickernel}
K^{\al, \be, \ga, \de}(x, y) = \mathfrak{C}(\al, \be, \ga, \de)\cdot \frac{\mathfrak{F}_1(x)\mathfrak{F}_0(y) - \mathfrak{F}_1(y)\mathfrak{F}_0(x)}{x - y}, \quad x, y\in\L,
\end{equation}
the constant is
\begin{equation}\label{Cabcd}
\mathfrak{C}(\al, \be, \ga, \de) := \frac{\thq(\ga\zm, \ga\zp, \de\zm, \de\zp)}{\zp\cdot\thq\left( \dfrac{\zm}{\zp}, \ga\de\zm\zp \right)}
\cdot \frac{\left( \dfrac{\al\be}{\ga\de}, \dfrac{\al\be}{q\ga\de}; q \right)_{\infty}}{\left( \dfrac{\al}{\ga}, \dfrac{\al}{\de}, \dfrac{\be}{\ga}, \dfrac{\be}{\de}, q, q; q \right)_{\infty}}
\end{equation}
and the functions $\mathfrak{F}_{0}(x), \mathfrak{F}_{1}(x)$ on $\L$ are
\begin{equation}\label{Fdef}
\mathfrak{F}_{\rr}(x) := \sqrt{|x| \dfrac{ (x\al,x\be;q)_\infty}{\thq(x\ga, x\de)}}  \cdot
(-x)^{1-\rr}\, \frac{\left(\dfrac{\be q^{\rr-1}}{\ga}, \dfrac{q^\rr}{\de x};q\right)_\infty}{\left( \dfrac{\al\be q^{2\rr-2}}{\ga\de} ; q\right)_{\infty}}
\cdot \, {}_2\phi_1\left(\begin{matrix}
\dfrac{\al q^{\rr-1}}\de, \, \dfrac q{\be x}\\
\dfrac{q^\rr}{\de x}\end{matrix}\Bigg| \frac{\be q^{\rr-1}}\ga\right), \ \rr = 0, 1.
\end{equation}
Above we used traditional $q$-calculus notation (\cite{GR}) for the $q$-Pochhammer symbols, theta functions and $q$-hypergeometric functions:
\begin{gather}
(z; q)_{\infty} := \prod_{i=1}^{\infty}{(1 - zq^{i - 1})}, \quad (z_1, \dots, z_m; q)_{\infty} := (z_1; q)_{\infty}\cdots (z_m; q)_{\infty},\nonumber\\
\thq(z) := (z, q/z; q)_{\infty}, \quad \thq(z_1, \dots, z_m) := \thq(z_1)\cdots \thq(z_m),\nonumber\\
_2\phi_1(a_1, a_2; b \mid z)  =  \, {}_2\phi_1\left(\begin{matrix} a_1, \, a_2 \\ b \end{matrix} \Bigg|  z  \right) :=
1 + \sum_{n = 1}^{\infty}{ z^n \prod_{i=1}^n{\frac{(1 - a_1q^{i-1})(1 - a_2q^{i-1})}{(1 - bq^{i-1})(1 - q^i)}} }.\label{def:hypergeo}
\end{gather}
A few remarks are in order:

\smallskip

1. The $q$-Pochhammer symbol $(z; q)_{\infty}$ is an entire function of $z$, whose set of zeroes is $q^{\Z_{\leq 0}}$.
Likewise, the theta function $\thq(z)$ is holomorphic on $\C^* = \C\setminus\{0\}$ and its zeroes are the points of $q^{\Z}$.
It follows that the denominator of $\mathfrak{C}(\al, \be, \ga, \de)$ does not vanish:
$\zm/\zp , \ga\de\zm\zp < 0$, so $\zm/\zp , \ga\de\zm\zp \notin q^{\Z}$, whereas if $\al/\ga \in q^{\Z_{\leq 0}}$, then $\al\be < q^2\ga\de$ is impossible (similarly for the other fractions $\al/\de, \be/\ga, \be/\de$).
Also, the denominator $\sqrt{\thq(x\ga, x\de)}\cdot (\al\be q^{2\rr - 2}/(\ga\de); q)_{\infty}$ of $\F_{\rr}(x)$ is nonzero for $x\notin\ga^{-1}q^{\Z} \cup \de^{-1}q^{\Z}$, in particular for $x\in\mathfrak{L}$.

\smallskip

2. The $q$-hypergeometric function $_2\phi_1(a_1, a_2; b \mid z)$ is analytic on the unit disk $|z| < 1$.
As a function of $z$ -- denote $F(z) := {}_2\phi_1(a_1, a_2; b \mid z)$ -- it satisfies the $q$-difference equation (see e.g. the survey \cite{K})
\begin{equation}\label{qdiff}
(b - a_1a_2qz) F(q^2z) + (-b-q+(a_1+a_2)qz)F(qz) + q(1 - z) F(z) = 0, \qquad |z| < 1.
\end{equation}
This implies that the value of $F(z)$ can be obtained as a linear combination of $F(qz)$ and $F(q^2z)$.
The expressions for the coefficients of the linear combination are rational functions with, at most, simple poles at $z = 1$.
Consequently the $q$-hypergeometric function can be analytically continued to a meromorphic function with simple poles at the points $1, q^{-1}, q^{-2}, \ldots$.
The analytic continuation will also be denoted by $_2\phi_1(a_1, a_2; b \mid z)$.
The meromorphicity of the $q$-hypergeometric function is not present when $q = 1$ --- in that case, the hypergeometric function $_2F_1(a_1, a_2; b \mid z)$ is often defined only on $\C\setminus [1, \infty)$.

\smallskip

3. The previous remark implies that $(z; q)_{\infty}\cdot{}_2\phi_1(a_1, a_2; b \mid z)$ is an entire function of $z$.
As a function of $b$, the $q$-hypergeometric function $_2\phi_1(a_1, a_2; b \mid z)$ is also meromorphic with simple poles at the points of $q^{\Z_{\leq 0}}$ and therefore $(b; q)_{\infty}\cdot{}_2\phi_1(a_1, a_2; b \mid z)$ is an entire function of $b$.
Both these points imply that $(b, z; q)_{\infty}\cdot {}_2\phi_1(a_1, a_2; b \mid z)$ is an entire function on $b$ and $z$.
As a result (see also the first remark), the right hand side of \eqref{Fdef} is well-defined whenever $x\notin \ga^{-1}q^{\Z} \cup \de^{-1}q^{\Z}\cup \{0\}$, i.e., the function $\F_{\rr}(x)$ of $x$ can be continued from $\L$ to $\C\setminus (\ga^{-1}q^{\Z} \cup \de^{-1}q^{\Z}\cup \{0\})$.

\smallskip

4. When $x = y$, we make sense of \eqref{basickernel} by using L'H\^opital's rule:
$$K^{\al, \be, \ga, \de}(x, x) = \mathfrak{C}(\al, \be, \ga, \de) \cdot (\mathfrak{F}'_1(x)\mathfrak{F}_0(x) - \mathfrak{F}_1(x)\mathfrak{F}'_0(x)), \qquad x\in\L,$$
where $\mathfrak{F}_{\rr}'$ is the derivative of $\mathfrak{F}_{\rr}$ (when $\mathfrak{F}_{\rr}$ is now seen as a function on $\C \setminus (\ga^{-1}q^{\Z} \cup \de^{-1}q^{\Z} \cup \{0\})$).

\smallskip

5. From the first and third remarks above, the functions $\F_0(x), \F_1(x)$ admit analytic continuations to the domains $\mathcal{D}_{\pm} := \C_{\pm} \setminus (\ga^{-1}q^{\Z} \cup \de^{-1}q^{\Z})$, where $\C_{\pm} := \{ z\in\C : \pm\Re z > 0 \}$ (replace $\sqrt{|x|}$ by $(\pm x)^{1/2}$, for $x\in\mathcal{D}_{\pm}$).
The previous remark and Cauchy's integral formula then show
\begin{equation}\label{basicdiagonal}
K^{\al, \be, \ga, \de}(x, x) = \oint_{|z - x| = \epsilon}{\frac{K^{\al,\be,\ga,\de}(z, x)}{z - x}dz}, \qquad x\in\L,
\end{equation}
for a sufficiently small $\epsilon > 0$ (depending on $x$).
If we let $\mathcal{D} := \mathcal{D}_+ \sqcup \mathcal{D}_- = \C \setminus (\ga^{-1}q^{\Z} \cup \de^{-1}q^{\Z} \cup \{ \Re z = 0\})$, then $K^{\al, \be, \ga, \de}(x, y)$ admits an analytic continuation to the domain $(x, y) \in \mathcal{D}^2$.

\medskip

The point processes $M^{\al, \be, \ga, \de}$ on $\L$ (or the corresponding probability measures on the space of point configurations on $\L$) are called the \emph{$q$--$zw$ measures}.
They are determinantal point processes with correlation kernels $K^{\al, \be, \ga, \de}(x, y)$.
We call $K^{\al, \be, \ga, \de}(x, y)$ the \emph{basic hypergeometric kernel}.

\section{The elliptic tail kernel}\label{sec:ell}

\subsection{Theta functions}
One of the basic properties of the theta function $\thq(z) = (z, q/z; q)_{\infty}$ is the quasi--periodicity
\begin{equation}\label{quasiperiodicity}
\thq(q^n z) = (-1)^n q^{-\frac{n(n-1)}{2}}z^{-n} \thq(z), \ n\in\Z,
\end{equation}
which we often use without mention.
They also satisfy the following fundamental identity.
\begin{lemma}\label{lem:qelliptic}
For any $X, Y, Z, W\in\C^*$,
\begin{multline}\label{sumthetas}
\thq(qYZ, \ Z/Y, \ qXW, \ W/X) - \thq(qYW, \ W/Y, \ qXZ, \ Z/X)\\
= -\frac{Z}{Y}\cdot\thq(qXY, \ Y/X, \ qZW, \ W/Z).
\end{multline}
\end{lemma}

\begin{proof}
This is known as the three-term relation due to Weierstrass, see e.g. \cite[(1.12)]{R}.
The formula in this reference is seen to be equivalent to \eqref{sumthetas} by use of the identities \eqref{quasiperiodicity} and $\thq(z) = \thq(q/z)$.
\end{proof}

\subsection{Definition of the elliptic tail kernel}

For any admissible pair $(\ga, \de)\in\C^2$, as in Definition \ref{admissiblepair}, let
\begin{equation}\label{Cconst}
C = C(\ga, \de) := \frac{\thq(\ga\zm, \ga \zp, \de \zm, \de \zp)}{\zp \cdot\thq\left( \dfrac{\zm}{\zp}, \ga\de\zm\zp \right)}
\cdot \frac{(\de - \ga)}{\ga\de\cdot \left( \dfrac{\de}{\ga}, \dfrac{\ga}{\de}, q, q; q \right)_{\infty}}.
\end{equation}

The \emph{elliptic tail kernel} associated to $(\ga, \de)$ is the function on $\L \times \L$ given by
\begin{equation}\label{thetakernel}
K^{\ga, \de}(x, y) := C\cdot\frac{P(x)Q(y) - Q(x)P(y)}{x - y}, \quad x, y\in\L,
\end{equation}
where
\begin{equation}\label{PQdef}
P(x) := \sqrt{|x|}\frac{\thq(x\de)}{\sqrt{\thq(x\ga, x\de)}},\quad
Q(x) := \sqrt{|x|}\frac{\thq(x\ga)}{\sqrt{\thq(x\ga, x\de)}}, \quad x\in\L.
\end{equation}
When $x = y$, the kernel is given by L'H\^opital's rule:
\begin{equation}\label{thetakernelequal}
K^{\ga, \de}(x, x) := C\cdot (P'(x)Q(x) - Q'(x)P(x)), \quad x\in\L.
\end{equation}
The formulas in \eqref{PQdef} give analytic continuations for the functions $P(x), Q(x)$ to the domains $\mathcal{D}_{\pm} := \C_{\pm} \setminus (\ga^{-1}q^{\Z} \cup \de^{-1}q^{\Z})$, where $\C_{\pm} := \{ z\in\C : \pm\Re z > 0 \}$ (replace $\sqrt{|x|}$ by $(\pm x)^{1/2}$ if $x\in\mathcal{D}_{\pm}$).
Let $\mathcal{D} := \mathcal{D}_+ \sqcup \mathcal{D}_- = \C \setminus (\ga^{-1}q^{\Z} \cup \de^{-1}q^{\Z} \cup \{ \Re z = 0\})$.
Note that $\L \subset \D$ because $(\ga, \de)$ is an admissible pair.
Thus \eqref{thetakernelequal} implies
\begin{equation}\label{ellipticdiagonal}
K^{\ga, \de}(x, x) = \oint_{|z - x| = \epsilon}{\frac{K^{\ga,\de}(z, x)}{z - x}dz}, \quad x\in\L,
\end{equation}
for a sufficiently small $\epsilon > 0$ (depending on $x$).
It follows that $K^{\ga, \de}(x, y)$ admits an analytic continuation to $\mathcal{D}^2$.

\subsection{The $q$-translation-invariance property}

Let $\epsilon: \L \rightarrow \{+1, -1\}$ be defined by $\epsilon(x) := 1$ if $x = \zp q^m$, and $\epsilon(x) := (-1)^n$ if $x = \zm q^n$.
Consider the gauge-transformed kernel
\begin{equation}\label{dfgauge}
\tilK^{\ga, \de}(x, y) := \epsilon(x)\epsilon(y)^{-1} K^{\ga, \de}(x, y) = \epsilon(x)\epsilon(y) K^{\ga, \de}(x, y).
\end{equation}

\begin{proposition}\label{Kperiodic}
The kernel $\tilK^{\ga, \de}$ is \emph{$q$-translation-invariant}, i.e.,
$\tilK^{\ga, \de}(qx, qy) = \tilK^{\ga, \de}(x, y)$, for any $x, y\in\L$.
\end{proposition}
\begin{proof}
From the definitions \eqref{thetakernel} and \eqref{dfgauge}, we can write
$$\tilK^{\ga, \de}(x, y) := C\cdot\frac{\tilP(x)\tilQ(y) - \tilQ(x)\tilP(y)}{x - y}, \quad x, y\in\L,$$
where
\begin{equation*}
\tilP(x) := \epsilon(x)\sqrt{|x|}\frac{\thq(x\de)}{\sqrt{\thq(x\ga, x\de)}},\quad
\tilQ(x) := \epsilon(x)\sqrt{|x|}\frac{\thq(x\ga)}{\sqrt{\thq(x\ga, x\de)}}, \quad x\in\L.
\end{equation*}
The quasi-periodicity property of the theta functions shows
$$\tilP(qx) = -\frac{|x|}{x} \epsilon(qx) \epsilon(x)^{-1} \frac{\sqrt{q\ga\de}}{\de} \tilP(x), \quad
\tilQ(qx) = -\frac{|x|}{x}\epsilon(qx) \epsilon(x)^{-1} \frac{\sqrt{q\ga\de}}{\ga} \tilQ(x), \quad x\in\L.$$
If $x>0$, then $\epsilon(x) = \epsilon(qx) = 1$, so $\tilP(qx) = -(\sqrt{q\ga\de}/\de) \tilP(x)$ and $\tilQ(qx) = -(\sqrt{q\ga\de}/\ga) \tilQ(x)$.
If $x<0$, then $\epsilon(qx) = -\epsilon(x)$, so again $\tilP(qx) = -(\sqrt{q\ga\de}/\de) \tilP(x)$ and $\tilQ(qx) = -(\sqrt{q\ga\de}/\ga) \tilQ(x)$.
It follows that $\tilK^{\ga, \de}(qx, qy) = \tilK^{\ga, \de}(x, y)$ whenever $x \neq y$ are points in $\L$.
It remains to prove the equality when $x = y$ is a point in $\L$.

Similarly as above, one has $K^{\ga, \de}(qx, qy) = K^{\ga, \de}(x, y)$ whenever $x \neq y$ are points in $\L$ of the same sign (note that this equality is for $K^{\ga, \de}$ and not $\tilK^{\ga, \de}$).
Since $K^{\ga, \de}(x, y)$ admits an analytic continuation to the $q$-invariant domain $\mathcal{D}_+^2 \sqcup \mathcal{D}_-^2 \subset \mathcal{D}^2$ and both $\D_-, \D_+$ are path-connected, then also $K^{\ga, \de}(qx, qx) = K^{\ga, \de}(x, x)$, for all $x\in\mathcal{D} = \mathcal{D}_+ \sqcup \mathcal{D}_-$, in particular for all $x\in\L$.
Hence, $\tilK^{\ga, \de}(qx, qx) = K^{\ga, \de}(qx, qx) = K^{\ga, \de}(x, x) = \tilK^{\ga, \de}(x, x)$, for all $x\in\L$.
\end{proof}

\subsection{Simplified formulas for the elliptic tail kernel}\label{app:elliptic}

We simplify $K^{\ga, \de}(x, y)$ when $x = \z_{\pm} q^{m}$, $y = \z_{\pm} q^{n}$, for $m, n \in \Z$.
The next two lemmas are simple consequences of the quasi-periodicity property of the theta function; we omit their proofs.
Recall that $C = C(\ga, \de)$ was defined in \eqref{Cconst}.

\begin{lemma}\label{ppsimplifying}
For any integers $m \neq n$,
\begin{align*}
K^{\ga, \de}(\zp q^m, \zp q^n) &= C (-1)^{m+n} \times \frac{\dfrac{\ga^m\de^n}{(\sqrt{\ga\de})^{m+n}} - \dfrac{\ga^n\de^m}{(\sqrt{\ga\de})^{m+n}}}{q^{\frac{m - n}{2}} - q^{\frac{n - m}{2}}},\\
K^{\ga, \de}(\zm q^m, \zm q^n) &= C \times \frac{\dfrac{\ga^n\de^m}{(\sqrt{\ga\de})^{m+n}} - \dfrac{\ga^m\de^n}{(\sqrt{\ga\de})^{m+n}}}{q^{\frac{m - n}{2}} - q^{\frac{n - m}{2}}}.
\end{align*}
\end{lemma}

\begin{lemma}\label{pmsimplifying}
For any $m, n\in\Z$,
\begin{multline*}
K^{\ga, \de}(\zp q^m, \zm q^n) = K^{\ga, \de}(\zm q^n, \zp q^m)\\
= \frac{C (-1)^{m}(\sqrt{\ga\de})^{-(m+n)}}{\sqrt{\thq(\zm\ga, \zm\de, \zp\ga, \zp\de)}} \cdot \frac{\ga^m\de^n \thq(\zm\ga, \zp\de) - \ga^n\de^m \thq(\zm\de, \zp\ga)}{|\zp/\zm|^{\frac{1}{2}}q^{\frac{m-n}{2}} + |\zm/\zp|^{\frac{1}{2}}q^{\frac{n-m}{2}}}.
\end{multline*}
\end{lemma}

\begin{lemma}\label{lem:zero}
For any $m\in\Z$,
\begin{align}
K^{\ga, \de}(\zp q^m, \zp q^m) = K^{\ga, \de}(\zp, \zp) &= C\zp \left\{ \de\frac{\thq'(\de\zp)}{\thq(\de\zp)} - \ga\frac{\thq'(\ga\zp)}{\thq(\ga\zp)} \right\},\label{pp0}\\
K^{\ga, \de}(\zm q^m, \zm q^m) = K^{\ga, \de}(\zm, \zm) &= C\zm \left\{ \ga\frac{\thq'(\ga\zm)}{\thq(\ga\zm)} - \de\frac{\thq'(\de\zm)}{\thq(\de\zm)} \right\}.\label{mm0}
\end{align}
\end{lemma}
\begin{proof}
Proposition \ref{Kperiodic} proves the first equalities in \eqref{pp0} and \eqref{mm0}.
Let us prove the second equality in \eqref{pp0}; the proof of the second equality in \eqref{mm0} is similar and we omit it.

At the points $(x, x) \in \L^2$ of the diagonal, the elliptic tail kernel is
$$K^{\ga, \de}(x, x) = C\cdot (Q(x)P'(x) - P(x)Q'(x)),$$
where $P(x) := \sqrt{|x|\thq(x\de, x\ga)}/\thq(x\ga)$, $Q(x) := \sqrt{|x|\thq(x\ga, x\de)}/\thq(x\de)$.
Then $P(x)^2 = |x|\dfrac{\thq(x\de)}{\thq(x\ga)}$.
When $x > 0$, we can take derivatives and obtain
\begin{equation*}
2P(x)P'(x) = x \frac{d}{dx}\frac{\thq(x\de)}{\thq(x\ga)} + \frac{\thq(x\de)}{\thq(x\ga)}
= x \frac{\de\thq(x\ga)\thq'(x\de) - \ga\thq(x\de)\thq'(x\ga)}{\thq(x\ga)^2} + \frac{\thq(x\de)}{\thq(x\ga)}.
\end{equation*}
Multiply by $Q(x)/(2P(x)) = \thq(x\ga)/(2\thq(x\de))$ to get
\begin{equation*}
Q(x)P'(x) = \frac{1}{2} \left\{ x \frac{\de\thq(x\ga)\thq'(x\de) - \ga\thq(x\de)\thq'(x\ga)}{\thq(x\ga, x\de)} + 1 \right\};
\end{equation*}
similarly $P(x)Q'(x)$ is given by the same formula with the swap $\ga \leftrightarrow \de$.
Hence
\begin{equation*}
K^{\ga, \de}(\zp, \zp) = C\cdot (Q(\zp)P'(\zp) - P(\zp)Q'(\zp))
= \left. Cx \cdot\frac{\de\thq(x\ga)\thq'(x\de) - \ga\thq(x\de)\thq'(x\ga)}{\thq(x\ga, x\de)} \right|_{x = \zp}
\end{equation*}
which proves \eqref{pp0}.
\end{proof}

Finally, we comment on the elliptic tail kernel with parameters $\de = \ga \in \R$ (part of the complementary series).
In the limit $\de\to\ga$, we have
$$
 \frac{(\de - \ga)}{\ga\de \cdot ({\de}/{\ga},{\ga}/{\de}; q)_{\infty}}
= \frac{1}{(\ga - \de) \cdot ({q\de}/{\ga},{q\ga}/{\de}; q)_{\infty}}
\sim \frac{1}{(\ga - \de)\cdot (q; q)^2_{\infty}}
$$
and
$$
\frac{\thq(x\de)\thq(y\ga) - \thq(x\ga)\thq(y\de)}{\ga - \de} \to y\thq(x\ga)\thq'(y\ga) - x\thq'(x\ga)\thq(y\ga).
$$

Then, when $\de=\ga$, the elliptic tail kernel can be simplified to
\begin{equation}\label{gaequalde}
K^{\ga, \ga}(x, y) =
\frac{\thq(\ga \z_-,\; \ga  \z_+)^2}{\z_+ \cdot (q; q)^{4}_{\infty} \cdot \thq\left(\dfrac{\z_-}{\z_+}, \; \ga^2 \z_-\z_+\right)}  \cdot \frac{ \sqrt{|x y|} }{x - y} \cdot \left( y\frac{\thq'(y\ga)}{\thq(y\ga)} - x\frac{\thq'(x\ga)}{\thq(x\ga)} \right)
\end{equation}
in the case $x \neq y$, and it is analytically continued according to L'H\^opital's rule in the case $x = y$.

\section{Limit from the basic hypergeometric kernel to the elliptic tail kernel}\label{sec:limit}

This section is devoted to proving the following limit.

\begin{theorem}\label{thm:mainlimit}
For any $x,y\in\L$, we have
\begin{equation*}
\lim_{M\to\infty} (\sgn(x)\sgn(y))^M K^{\al, \be, \ga, \de}(q^Mx,q^My) = K^{\ga, \de}(x, y),
\end{equation*}
where $K^{\al, \be, \ga, \de}$ is defined by \eqref{basickernel} and $K^{\ga, \de}$ is defined by \eqref{thetakernel}.
\end{theorem}

\begin{proof}
We first transform the functions $\F_\rr(x)$ so that they are well-suited for the $x\to 0$ limit.

The Heine transformation formula for $_2\phi_1$ (cf.\ \cite[Sec. 1.4]{GR}) yields
$$
 _2\phi_1\left(\begin{matrix}
A, \, B\\
C\end{matrix}\Bigg| z\right)=\frac{(B,Az;q)_\infty}{(C,z;q)_\infty} \,
_2\phi_1\left(\begin{matrix}
C/B, \, z\\
Az\end{matrix}\Bigg| B\right).
$$
We use it with
$$
 A=\dfrac{\al q^{\rr-1}}\de,\quad B=\dfrac q{\be x},\quad C=\dfrac{q^\rr}{\de
 x},\quad  z=\frac{\be q^{\rr-1}}{\ga},
$$
to get

\begin{multline*}
\F_\rr(x) = \sqrt{|x| \dfrac{ (x\al,x\be;q)_\infty}{\thq(x\ga, x\de)}}  \cdot
(-x)^{1-\rr}\, \frac{\left(\dfrac{\be q^{\rr-1}}\ga, \dfrac{q^\rr}{\de x};q\right)_\infty}{\left( \dfrac{\al\be q^{2\rr-2}}{\ga\de} ; q\right)_{\infty}}
\cdot \, {}_2\phi_1\left(\begin{matrix}
\dfrac{\al q^{\rr-1}}\de, \, \dfrac q{\be x}\\
\dfrac{q^\rr}{\de x}\end{matrix}\Bigg| \frac{\be q^{\rr-1}}\ga\right)
\\=
\sqrt{|x| \dfrac{ (x\al,x\be;q)_\infty}{\thq(x\ga, x\de)}}  \cdot
(-x)^{1-\rr}\, \cdot \, \left(\dfrac q{\be x}; q \right)_\infty \cdot  \,_2\phi_1\left(\begin{matrix}
\dfrac{\be q^{\rr-1}}{\de}, \, \dfrac{\be q^{\rr-1}}\ga\\
\dfrac{\al\be}{\ga \de} q^{2\rr-2}\end{matrix}\Bigg| \dfrac q{\be x}\right).
\end{multline*}

Further, recall Watson's formula (see \cite[(4.3.2)]{GR})
\begin{multline*}
 _2\phi_1\left(\begin{matrix}
A, \, B\\
C\end{matrix}\Bigg| z\right) =\frac{(B,C/A,Az, q/Az; q)_\infty}{(C,B/A,z,q/z; q)_\infty}
\,_2\phi_1\left(\begin{matrix}
A, \, Aq/C\\
Aq/B\end{matrix}\Bigg| Cq/ABz\right)
\\
+\frac{(A,C/B,Bz, q/Bz; q)_\infty}{(C,A/B,z,q/z; q)_\infty} \,_2\phi_1\left(\begin{matrix}
B, \, Bq/C\\
Bq/A\end{matrix}\Bigg| Cq/ABz\right)
\end{multline*}
and apply it with
$$
 A=\dfrac{\be q^{\rr-1}}{\de},\quad
 B= \dfrac{\be q^{\rr-1}}{\ga},\quad
 C= \dfrac{\al\be q^{2\rr-2}}{\ga \de},\quad
 z=\dfrac {q}{\be x},
$$
to get
\begin{multline*}
\F_\rr(x)= \sqrt{|x| \dfrac{ (x\al,x\be;q)_\infty}{\thq(x\ga, x\de)}} \cdot
(-x)^{1-\rr}\,  \frac{1}{\left( \beta x, \dfrac{\al\be q^{2\rr-2}}{\ga\de}; q \right)_\infty} \\
\times \Biggl[ \frac{\left(\dfrac{\be q^{\rr-1}}{\ga},\dfrac{\al q^{\rr-1}}{\ga} ; q\right)_\infty\thq\left(x\de
q^{1-\rr}\right)}{\left(\dfrac{\de}{\ga}; q\right)_\infty}
\,_2\phi_1\left(\begin{matrix}
\dfrac{\be q^{\rr-1}}{\de}, \, \dfrac{\ga q^{2-\rr}}{\al}\\
\dfrac{\ga q}{\de}\end{matrix}\Bigg|\alpha x\right)\\
+  \{  \text{same expression after the swap } \ga \leftrightarrow \de \} \Biggr].
\end{multline*}

In the formula above, do the change of variables $x \mapsto q^M x$ (later we send $M\to\infty$).
The quasi-periodicity of the theta function implies
$$
\thq(A q^M x)= (-Ax)^{-M} q^{-M(M-1)/2} \thq(Ax),
$$
thus
\begin{multline*}
\F_\rr(q^Mx)= \sqrt{|q^M x| \dfrac{(q^Mx\al,q^Mx\be;q)_\infty}{q^{-M(M-1)}(x^2\ga\de)^{-M}\thq(x\ga, x\de)}}
\cdot
(-x q^M)^{1-\rr}\,  \frac{1}{\left( \beta q^Mx, \dfrac{\al\be q^{2\rr-2}}{\ga\de}; q \right)_\infty} \cdot \,  \\
\times \Biggl[ \frac{\left(\dfrac{\be q^{\rr-1}}{\ga},\dfrac{\al q^{\rr-1}}{\ga}; q \right)_\infty \thq\left(x\de
q^{1-\rr}\right)}{\left(\dfrac{\de}{\ga} ; q\right)_\infty} q^{-M(M-1)/2}(-x\de
q^{1-\rr})^{-M} \,_2\phi_1\left(\begin{matrix}
\dfrac{\be q^{\rr-1}}{\de}, \, \dfrac{\ga q^{2-\rr}}{\al}\\
\dfrac{\ga q}{\de}\end{matrix}\Bigg|\alpha q^Mx\right)\\
+ \{  \text{same expression after the swap } \ga \leftrightarrow \de \} \Biggr]
\\
= (-\sgn(x))^M \sqrt{|q^M x| \dfrac{ (q^Mx\al,q^Mx\be; q)_\infty}{\thq(x\ga, x\de)}} \cdot
(-x)^{1-\rr}\,  \frac{1}{\left( \be q^Mx, \dfrac{\al\be q^{2\rr-2}}{\ga\de}; q \right)_\infty} \cdot \,  \\
\times \Biggl[ \frac{\left(\dfrac{\be q^{\rr-1}}{\ga},\dfrac{\al q^{\rr-1}}{\ga}; q\right)_\infty\thq\left(x\de
q^{1-\rr}\right)}{\left(\dfrac{\de}{\ga}; q\right)_\infty} \left(\frac{\ga}{\de}
\right)^{M/2} \,_2\phi_1\left(\begin{matrix}
\dfrac{\be q^{\rr-1}}{\de}, \, \dfrac{\ga q^{2-\rr}}{\al}\\
\dfrac{\ga}{\de}q\end{matrix}\Bigg|\alpha q^Mx\right)
\\
+ \{  \text{same expression after the swap } \ga \leftrightarrow \de \} \Biggr].
\end{multline*}
For the correlation kernel, we only need $\F_0$ and $\F_1$. Let us analyze them more
carefully, keeping only the first term of the $M\to\infty$ asymptotics (for $\F_0$, we also need $\thq(x\ga q) = (-x\ga)^{-1}\thq(x\ga)$ and $\thq(x\de q) = (-x\de)^{-1}\thq(x\de)$):
\begin{multline*}
\F_0(q^Mx)=  (-\sgn(x))^M \sqrt{|q^M x| \dfrac{(q^Mx\al, q^Mx\be;q)_\infty}{\thq(x\ga, x\de)}} \cdot
\frac{1}{\left( \be q^Mx, \dfrac{\al\be}{q^2\ga\de}; q \right)_\infty} \cdot \,  \\
\times \Biggl[ \de^{-1} \left( 1 - \frac{\be q^{-1}}{\ga} \right) \left(1 - \frac{\al q^{-1}}{\ga}\right)\cdot
\frac{\left(\dfrac{\be}{\ga},\dfrac{\al}{\ga} ; q\right)_\infty \thq\left(x\de
\right)}{\left(\dfrac{\de}{\ga}; q\right)_\infty} \left(\frac{\ga}{\de} \right)^{M/2}
\left(1+O(q^{M})\right)\\
+ \{  \text{same expression after the swap } \ga \leftrightarrow \de \} \Biggr],
\end{multline*}

\begin{multline*}
\F_1(q^Mx)= (-\sgn(x))^M \sqrt{|q^M x| \dfrac{(q^Mx\al,q^Mx\be;q)_\infty}{\thq(x\ga, x\de)}} \cdot
  \frac{1}{\left( \be q^Mx, \dfrac{\al\be}{\ga\de}; q \right)_\infty} \cdot \,  \\
\times \Biggl[ \frac{\left(\dfrac{\be}{\ga},\dfrac{\al}{\ga}; q\right)_\infty
\thq\left(x\de \right)}{\left(\dfrac{\de}{\ga}; q\right)_\infty}
\left(\frac{\ga}{\de} \right)^{M/2} \left(1+O(q^{M})\right)
+ \{  \text{same expression after the swap } \ga \leftrightarrow \de \} \Biggr].
\end{multline*}

Then
\begin{gather*}
K^{\al, \be, \ga, \de}(q^Mx,q^My) = \mathfrak{C}(\al, \be, \ga, \de) \cdot \frac{\F_1(q^Mx)\F_0(q^My)-\F_1(q^My)\F_0(q^Mx)}{q^Mx-q^My}\\
= \mathfrak{C}(\al, \be, \ga, \de) \cdot \frac{\left( \dfrac{\al}{\ga}, \dfrac{\al}{\de}, \dfrac{\be}{\ga}, \dfrac{\be}{\de}; q\right)_\infty}
{\left(\dfrac{\al\be}{\ga\de}, \dfrac{\al\be}{q^2\ga\de}, \dfrac{\de}{\ga},\dfrac{\ga}{\de}; q\right)_\infty} \cdot (\sgn(x)\sgn(y))^M\\
\times \sqrt{|x| \dfrac{(q^Mx\al,q^Mx\be;q)_\infty}{\thq(x\ga, x\de)}} \cdot
\frac{1}{(\beta q^Mx; q)_\infty}\sqrt{|y| \dfrac{(q^My\al,q^My\be;q)_\infty}{\thq(y\ga, y\de)}} \cdot
\frac{1}{(\beta q^My; q)_\infty}\\
\times \Biggl[ \left( \ga^{-1} \left( 1 - \frac{\be q^{-1}}{\de} \right) \left( 1 - \frac{\al q^{-1}}{\de} \right) -
\de^{-1} \left( 1 - \frac{\be q^{-1}}{\ga} \right) \left( 1 - \frac{\al q^{-1}}{\ga} \right)  \right)
\frac{\thq(x\de)\thq(y\ga) - \thq(x\ga)\thq(y\de)}{x - y}\\
+ O\left( q^M(\gamma/\delta)^{M/2} + q^M(\delta/\gamma)^{M/2} \right) \Biggr]\\
\sim \mathfrak{C}(\al, \be, \ga, \de)\cdot
 \frac{\left( \dfrac{\al}{\ga}, \dfrac{\al}{\de}, \dfrac{\be}{\ga}, \dfrac{\be}{\de}; q\right)_\infty}
{\left(\dfrac{\al\be}{\ga\de}, \dfrac{\al\be}{q^2\ga\de}, \dfrac{\de}{\ga},\dfrac{\ga}{\de}; q\right)_\infty}
\cdot \frac{(\de - \ga)(\ga\de q^2 - \al\be)}{\ga^2\de^2q^2}\\
\times (\sgn(x)\sgn(y))^M
 \sqrt{\dfrac{|x|}{\thq(x\ga, x\de)}}
  \sqrt{\dfrac{|y|}{\thq(y\ga, y\de)}}
\cdot \frac{\thq(x\de)\thq(y\ga)-\thq(x\ga)\thq(y\de)}{x - y},
\end{gather*}
where we denoted $A \sim B$ to mean $\lim_{q \to 1}{|A - B|} = 0$.
Note that we used $O(q^M(\ga/\de)^{M/2} + q^M(\de/\ga)^{M/2}) = o(1)$, as $M \to \infty$, which is a consequence of the fact that if $(\ga, \de)$ is an admissible pair, then $|q^2\ga/\de|,\ |q^2\de/\ga| \in (0, 1)$.

Plugging \eqref{Cabcd} into the estimate above gives
\begin{multline*}
(\sgn(x)\sgn(y))^M K^{\al, \be, \ga, \de}(q^Mx,q^My) \sim \frac{\thq(\ga\zm, \ga\zp, \de\zm, \de\zp)}{\zp\cdot \thq\left( \dfrac{\zm}{\zp}, \ga\de\zm\zp \right)\cdot (q; q)_{\infty}^2} \cdot \frac{1}{\left( \dfrac{\de}{\ga}, \dfrac{\ga}{\de}; q \right)_{\infty}}
\cdot \frac{(\de - \ga)}{\ga\de}\\
\times \sqrt{\dfrac{|x|}{\thq(x\ga, x\de)}}
  \sqrt{\dfrac{|y|}{\thq(y\ga, y\de)}}
\cdot \frac{\thq(x\de)\thq(y\ga)-\thq(x\ga)\thq(y\de)}{x - y},
\end{multline*}
and the expression above is exactly the right hand side of \eqref{thetakernel}.
Thus we have shown the desired limit for any $x \neq y$ in $\L$.
Recall that both kernels $K^{\al, \be, \ga, \de}(x, y)$ and $K^{\ga, \de}(x, y)$ admit analytic continuations to $\mathcal{D}^2$, where $\D = \C \setminus (\ga^{-1}q^{\Z} \cup \de^{-1}q^{\Z} \cup \{ \Re z = 0\})$.
Such analytic continuations are given by \eqref{basickernel} and \eqref{thetakernel} when $x \neq y$.
Then the estimates above actually show
\begin{equation*}
\lim_{M\to\infty} (\sgn(x)\sgn(y))^M K^{\al, \be, \ga, \de}(q^Mx,q^My) = K^{\ga, \de}(x, y),
\end{equation*}
for any $x \neq y$ in $\D$, where $\sgn(x) := \sgn(\Re x)$.
Moreover, the limit is uniform for $(x, y)$ varying over compact subsets of $\D^2 \setminus \{ (z, z) : z\in\D \}$.
Then the integral representations \eqref{basicdiagonal} and \eqref{ellipticdiagonal} imply that the desired limit also holds for $x = y$.
\end{proof}

\section{An application: The absence of atoms in the $q$--$zw$ measures}\label{sec:diffusity}

\subsection{A dichotomy for determinantal measures}

Let $\X$ be a countable set and $\Om$ be the set of all subsets of $\X$.
Observe that any $\om\in\Om$ can be interpreted as a $\{0, 1\}$-valued function on $\X$ if we identify a subset with its indicator function. As a result, $\Om$ is in bijection with $\{0, 1\}^{\X}$ and we can equip it with the product topology, so that it becomes a compact space.
Let $\mathcal P(\Om)$ denote the space of Borel probability measures on $\Om$. Any measure $M\in\mathcal P(\Om)$ is uniquely determined by its \emph{correlation functions} $\rho_1,\rho_2,\dots$. Here
$$
\rho_n(x_1,\dots,x_n):=M(\{\om\in\Om: \{x_1, \ldots, x_n\} \subseteq \om \}), \qquad n=1,2,\dots,
$$
where $x_1,\dots,x_n$ are pairwise distinct points of $\X$. $M$ is said to be a \emph{determinantal measure} if there exists a  complex-valued function $K(x,y)$ on $\X\times\X$ such that
$$
\rho_n(x_1,\dots,x_n)=\det[K(x_i,x_j)]_{i,j=1}^n, \qquad n=1,2,\dots\,.
$$
In particular, $\rho_1(x)=K(x,x)$. Any such function is called a \emph{correlation kernel} of $M$.

Let $\ell^2(\X)$ be the complex Hilbert space formed by complex-valued, square-summable functions on $\X$ and $\{e_x: x\in\X\}$ be its natural orthonormal basis formed by the delta functions. Given a bounded operator $K$ on $\ell^2(\X)$, we set
$$
K(x,y):=(Ke_y,e_x), \qquad x,y\in\X,
$$
and call $K(x,y)$ the kernel of $K$.

Suppose that $K$ is a \emph{positive contraction}, meaning that $K=K^*$ and $0\le K\le1$. Then $K$ gives rise to a determinantal measure $M^K\in\mathcal P(\Om)$, see e.g. \cite[Sec. 8]{Lyons}, \cite[Thm. 2.1]{ST}. Namely, the kernel of $K$  serves as a correlation kernel for  $M^K$.

Recall that a measure is said to be \emph{diffuse} if it has no atoms.

\begin{theorem}\label{thm:generaldiffuse}
Let $M\in\mathcal P(\Om)$ be a determinantal measure defined by a positive contraction on $\ell^2(\X)$. Let $\rho_1(x)$ be the first correlation function of $M$ and set
$$
\rho^*_1(x):=\min(\rho_1(x), 1-\rho_1(x)).
$$

Then the following dichotomy holds: $M$ is either diffuse or purely atomic, depending on whether the series $
\sum_{x\in\X}\rho^*_1(x)$ diverges or converges, respectively.
\end{theorem}

We give the proof after a little preparation. We need the following elementary lemma.

\begin{lemma}\label{boundproducts}
If $\AA = [\AA(i,j)]$ is a matrix of finite size, $\AA=\AA^*\ge0$, then
$$
\det \AA\le \prod_i \AA(i,i).
$$
\end{lemma}

\begin{proof}[Proof of the lemma]
Use induction on $N$, the size of $\AA$. Write $\AA$ in the block form corresponding to the partition $N = (N-1) + 1$:
$$
\AA=\begin{bmatrix}  A & B \\ B^* & D  \end{bmatrix}.
$$

Assume first that $A$ is nonsingular. Then $A>0$, $A$ is invertible, and we may write
$$
\det \AA=\det A \cdot (D - B^* A^{-1}B).
$$

Observe that $D-B^* A^{-1}B\le D$, since $A^{-1}$ is positive-definite.
Hence $\det\AA\le \det A \cdot D$, and we may apply induction.

If $A$ is singular, then we apply the above argument to $\AA+\epsi 1$ with small $\epsi>0$ and pass to the limit as $\epsi\to0$.
\end{proof}

We will also use the particle-hole involution $\om\mapsto \om^\circ$, where $\om^{\circ}:=\X\setminus\om$. It induces an involutive transformation $M\mapsto M^\circ$ of the space $\mathcal P(\Om)$. Note that if $M=M^K$, where $K$ is a positive contraction, then $1-K$ is a positive contraction too, and we have $M^\circ=M^{1-K}$, see \cite[Appendix \S A.3]{BOO}.

A trivial but important observation is that the particle-hole involution leaves the function $\rho_1^*(x)$ invariant.

Finally, suppose that $\X'$ is a subset of $\X$ and let $\Om'$ be the set of subsets of $\X'$ equipped with the product topology (after identifying $\Om'$ with $\{0, 1\}^{\X'}$).
The correspondence $\om\mapsto\om\cap\X'$ defines a projection $\Om\to \Om'$ and hence a map $\mathcal P(\Om)\to\mathcal P(\Om')$. If $M\in\mathcal P(\Om)$ is a determinantal measure, then its pushforward $M'$ under that map is a determinantal measure too, and if $M=M^K$ for a positive contraction, then $M'$ has a similar form, with the positive contraction $K'$ on $\ell^2(\X')$ whose kernel is the restriction of the kernel $K(x,y)$ to $\X'\times\X'$.  In particular, the function $\rho^*_1$ is simply restricted to $\X'$.

\begin{proof}[Proof of Theorem \ref{thm:generaldiffuse}]
1. We assume that $\sum_{x\in\X}\rho^*_1(x)=\infty$ and prove that $M$ is diffuse.
This means that $M$ assigns mass $0$ to any singleton $\{\om\}$.
We will first prove this for $\omega = \X$.
For any $n$-point subset $X=\{x_1,\dots,x_n\}\subset\X$ we have
$$
M(\{\X\})\le\rho_n(x_1,\dots,x_n).
$$
By Lemma \ref{boundproducts},
$$
\rho_n(x_1,\dots,x_n)\le\prod_{i=1}^n{\rho_1(x_i)}=\prod_{x\in X}\rho_1(x).
$$
Therefore, for any finite subset $X \subset \X$,
$$
M(\{\X\})\le\prod_{x\in X}\rho_1(x).
$$

On the other hand, for any $x\in\X$,
$$
\rho_1(x)\le 1-\rho^*_1(x).
$$
It follows that
$$
M(\{\X\})\le\prod_{x\in X}(1-\rho^*_1(x))
$$
for any finite $X$. Since $\sum_{x\in\X}\rho^*_1(x)=\infty$, the right-hand side can be made arbitrarily small with an appropriate choice of $X$. We conclude that $M(\{\X\})=0$, as desired.

Now let us consider the general case $\om\in\Om$.
Set
$$
\X^0:=\{x\in\X: x\notin\om\}, \qquad \X^1:=\{x\in\X: x\in\om\}.
$$
As $\X = \X^0 \sqcup \X^1$ and $\sum_{x\in\X}{\rho_1^*(x)} = \infty$, then $\sum_{x\in\X^0}\rho^*_1(x)=\infty$, or $\sum_{x\in\X^1}\rho^*_1(x)=\infty$ (or both).
Examine first the case when $\sum_{x\in\X^1}\rho^*_1(x)=\infty$. Then let $\X':=\X^1$ and $M'$ be the pushforward of $M$ under the projection $\Om\to\Om'$ defined above. We have $M(\{\om\})\le M'(\{\om'\})$, where $\om' = \X'$, so it suffices to prove $M'(\{\om'\}) = 0$.
This reduces the statement to the case $\om = \X$ above.

Now examine the case when $\sum_{x\in\X^0}\rho^*_1(x) = \infty$. Let $\X':=\X^0$ and, again, let $M'$ be the pushforward of $M$ under the projection $\Om\to\Om'$, we are now reduced to the case $\om=\emptyset$.
Then we can perform the particle-hole involution and use the invariance of $\rho^*_1$ to reduce the desired statement to the known case $\om = \X$.

2. Next, we assume that $\sum_{x\in\X}\rho^*_1(x)<\infty$ and prove that $M$ is purely atomic. We have $\X=\X_0\sqcup\X_1$, where
$$
\X_0:=\{x\in\X: \rho_1(x)\le\tfrac12\}, \qquad \X_1:=\{x\in\X: \rho_1(x)>\tfrac12\}.
$$
For $\om\in\Om$, let $\om\triangle\X_1$ denote the symmetric difference of $\om$ and $\X_1$. We set
$$
\Om^*:=\{\om\in\Om: |\om\triangle \X_1|<\infty\}
$$
and note that $\Om^*$ is a countable subset of $\Om$. We are going to show that $M$ is concentrated on $\Om^*$, which will imply that $M$ is purely atomic.
To do this, we treat $\om\in\Om$ as the random element distributed according to $M$. For $x\in\X_0$, let $E_x$ be the event that $x\in\om$, whereas if $x\in\X_1$, let $E_x$ be the event that $x\notin\om$. Then $\om\in\Om^*$ precisely means that only finitely many events $E_x$ occur.

On the other hand, the probability of $E_x$ is $\rho^*_1(x)$. Thus, the sum of all these probabilities is finite. Applying the Borel--Cantelli lemma, we obtain that $\om\in\Om^*$ with probability $1$, which is equivalent to the desired claim.
\end{proof}

\subsection{The $q$--$zw$ measures are diffuse}

\begin{theorem}\label{thm:diffuse}
Let $(\al, \be, \ga, \de) \in \C^4$ be any quadruple of admissible parameters.
The corresponding $q$--$zw$ measure is diffuse.
\end{theorem}

\begin{proof}

We claim that $0 < K^{\ga, \de}(\zp, \zp) < 1$, and similarly for $K^{\ga, \de}(\zm, \zm)$.
Before proving the claim, let us deduce Theorem \ref{thm:diffuse} from it.
Theorem \ref{thm:mainlimit} shows $K^{\al, \be, \ga, \de}(x,x) \to K^{\ga, \de}(\z_{\pm}, \z_{\pm})$, as $x \to 0^{\pm}$ in $\L$, thus the density $\rho_1^{\al, \be, \ga, \de}(x) = K^{\al, \be, \ga, \de}(x, x)$, when $|x|$ is small, is uniformly bounded away from $0$ and $1$.
Consequently, Theorem \ref{thm:generaldiffuse} shows that $M^{\al, \be, \ga, \de}$ is diffuse.

Next let us prove the claim.
Since $K^{\ga, \de}(\zp, \zp)$ and $K^{\ga, \de}(\zm, \zm)$ are probabilities, then the claim would be contradicted if and only if $\{ K^{\ga, \de}(\zp, \zp), K^{\ga, \de}(\zm, \zm) \} \cap \{0, 1\} \neq \emptyset$.
Let us focus on proving $K^{\ga, \de}(\zp, \zp) \neq 1$, as the proof of the other three statements is the same.
Assume $K^{\ga, \de}(\zp, \zp) = 1$.
Then $M^{\ga, \de}$--almost surely, a random point configuration contains $\zp$.
Since $M^{\ga, \de}$ is stationary, a random configuration contains all points to the right of the origin and so all the corresponding correlation functions equal $1$, in particular
$$\det\left( {\begin{array}{cc}  K^{\ga, \de}(\zp q^m, \zp q^m) & K^{\ga, \de}(\zp q^m, \zp q^n) \\[0.3ex]
K^{\ga, \de}(\zp q^n, \zp q^m) & K^{\ga, \de}(\zp q^n, \zp q^n) \\  \end{array} } \right) = 1, \text{ whenever }m \neq n.$$
The matrix above is symmetric and its diagonal entries are equal to $1$, therefore
$$K^{\ga, \de}(\zp q^m, \zp q^n) = K^{\ga, \de}(\zp q^n, \zp q^m) = 0.$$
Apply the previous equation to $(m, n) = (2, 0)$: Lemma \ref{ppsimplifying} gives
\begin{equation}\label{final_zero}
K^{\ga, \de}(\zp q^2, \zp) = C(\ga, \de) \times \frac{\ga/\de - \de/\ga}{q - q^{-1}} = 0.
\end{equation}
One verifies that for any admissible pair $(\ga, \de)$ with $\ga \neq \de$, one has $C(\ga, \de) \neq 0$ and $\ga/\de - \de/\ga \neq 0$.
Thus we have reached a contradiction.
In the special case $\ga = \de$, we need the formulas at the end of Section \ref{app:elliptic} (see \eqref{gaequalde}); then equation \eqref{final_zero} becomes
$$K^{\ga, \ga}(\zp q^2, \zp) = \frac{2}{\zp\ga (q - q^{-1})} \cdot \frac{\thq(\ga\zm, \ga\zp)^2}{\thq(\zm/\zp, \ga^2\zm\zp)\cdot (q; q)_{\infty}^4} = 0.$$
This implies that $\ga\in\zm^{-1}q^{\Z}$ or $\ga\in\zp^{-1}q^{\Z}$, which again is impossible.
\end{proof}

\section{The projection property of the elliptic tail kernel}\label{sec:projection}

\subsection{The projection property}

Let $(\ga, \de)$ be an admissible pair and let $K^{\ga, \de}$ be the operator on $\ell^2(\L)$ with kernel $K^{\ga, \de}(x, y)$, that is,
$$(K^{\ga, \de}f)(x) := \sum_{y\in\L}{K^{\ga, \de}(x, y)f(y)}, \qquad f\in\ell^2(\L), \quad x\in\L.$$
The main theorem of this section is the following.

\begin{theorem}\label{thm:projection}
The operator $K^{\ga, \de}$ on $\ell^2(\L)$ is a projection operator, i.e., $K^{\ga, \de} = (K^{\ga, \de})^* = (K^{\ga, \de})^2$.
\end{theorem}

In this section, it will be more convenient to use the gauge-transformed kernel $\tilK^{\ga, \de}(x, y) = \epsilon(x)\epsilon(y)K^{\ga, \de}(x, y)$ defined in \eqref{dfgauge} because of the $q$-translation-invariance property $\tilK^{\ga, \de}(qx, qy) = \tilK^{\ga, \de}(x, y)$ of Proposition \ref{Kperiodic}.
Clearly, the corresponding operator $\tilK^{\ga, \de}$ is a projection operator if and only if $K^{\ga, \de}$ is a projection operator.

Let us consider the $2\times 2$ matrix-valued kernel $\K(m, n) = \K^{\ga, \de}(m, n)$ on $\Z$, given by
\begin{equation*}
\K(m, n) := \begin{pmatrix} \tilK^{\ga, \de}(\zp q^m, \zp q^n) & \tilK^{\ga, \de}(\zp q^m, \zm q^n) \\[0.5ex]
\tilK^{\ga, \de}(\zm q^m, \zp q^n) & \tilK^{\ga, \de}(\zm q^m, \zm q^n) \end{pmatrix},\quad m, n\in\Z.
\end{equation*}
The space $\ell^2(\Z; \C^2)$ (Hilbert space of $\C^2$-valued, square-summable sequences) can be naturally identified with $\ell^2(\L)$ --- under this identification, $\tilK^{\ga, \de}$ becomes the operator with kernel $\K(m, n)$.
We go a step further.
Proposition \ref{Kperiodic} implies the translation--invariance property: $\K(m, n) = \K(m+1, n+1)$.
This suggests to look at the \emph{Fourier transform} $\hatK = \hatK^{\ga, \de}$ of the function $\K(m, 0)$, $m\in\Z$.
By definition, $\hatK$ is a $2\pi$-periodic function on $\R$, which is $2\times 2$ matrix-valued and given by
\begin{equation}\label{matrixFourier}
\hatK(\eta) := \begin{pmatrix} \hatK_{+, +}(\eta) & \hatK_{+, -}(\eta) \\[0.5ex]
\hatK_{-, +}(\eta) & \hatK_{-, -}(\eta) \end{pmatrix}, \quad \eta\in\R,
\end{equation}
\begin{equation}\label{defFourier}
\hatK_{\epsilon_1, \epsilon_2}(\eta)
:=\sum_{m\in\Z}{e^{i \eta m} \tilK^{\ga, \de}(\zeta_{\epsilon_1}q^m, \zeta_{\epsilon_2})}, \quad \epsilon_1, \epsilon_2 \in \{+, -\}.
\end{equation}

The important point for us is the following lemma, whose proof essentially follows by definition of the Fourier transform.

\begin{lemma}\label{proj_transform}
A translation--invariant operator $\K$ on $\ell^2(\Z; \C^2)$ is a projection operator if and only if its Fourier transform $\hatK(\eta)$ is a projection matrix, for any $\eta\in\R$.
\end{lemma}

In view of Lemma \ref{proj_transform}, a proof of Theorem \ref{thm:projection} will be furnished by the verification that $\hatK(\eta)$ is a projection matrix, for any $\eta\in\R$.
The latter will be a consequence of the following proposition, whose proof is given in Section \ref{sec:prooffourier}, after some preparations.

\begin{proposition}\label{prop:fourier}
For any $\eta\in\R$, the Fourier transform $\hatK(\eta)$, defined by \eqref{matrixFourier}--\eqref{defFourier}, is given by
\begingroup
\addtolength{\jot}{0.5em}
\begin{align}
\hatK_{+, +}(\eta) &= \frac{q\cdot\thq(\ga\zm, \de\zm)}{\ga\de\zp^2\cdot\thq(\zm/\zp, \ga\de\zm\zp)} \frac{\thq(-e^{i\eta}\zp\sqrt{\ga\de/q}, -e^{-i\eta}\zp\sqrt{\ga\de/q})}{\thq(-e^{i\eta}\sqrt{q\ga\de}/\ga, -e^{i\eta}\sqrt{q\ga\de}/\de)},\label{ppthm}\\
\hatK_{+, -}(\eta) &= -\frac{q\sqrt{\thq(\ga\zm, \de\zm, \ga\zp, \de\zp)}}{\ga\de\zp\sqrt{|\zm\zp|}\cdot\thq(\zm/\zp, \ga\de\zm\zp)} \frac{\thq(-e^{i\eta}\zp\sqrt{\ga\de/q}, -e^{-i\eta}\zm\sqrt{\ga\de/q})}{\thq(-e^{i\eta}\sqrt{q\ga\de}/\ga, -e^{i\eta}\sqrt{q\ga\de}/\de)},\label{pmthm}
\end{align}
\begin{align}
\hatK_{-, +}(\eta) &= -\frac{q\sqrt{\thq(\ga\zm, \de\zm, \ga\zp, \de\zp)}}{\ga\de\zp\sqrt{|\zm\zp|}\cdot\thq(\zm/\zp, \ga\de\zm\zp)} \frac{\thq(-e^{i\eta}\zm\sqrt{\ga\de/q}, -e^{-i\eta}\zp\sqrt{\ga\de/q})}{\thq(-e^{i\eta}\sqrt{q\ga\de}/\ga, -e^{i\eta}\sqrt{q\ga\de}/\de)},\label{mpthm}\\
\hatK_{-, -}(\eta) &= \frac{q\cdot\thq(\ga\zp, \de\zp)}{\ga\de|\zm\zp|\cdot\thq(\zm/\zp, \ga\de\zm\zp)} \frac{\thq(-e^{i\eta}\zm\sqrt{\ga\de/q}, -e^{-i\eta}\zm\sqrt{\ga\de/q})}{\thq(-e^{i\eta}\sqrt{q\ga\de}/\ga, -e^{i\eta}\sqrt{q\ga\de}/\de)}.\label{mmthm}
\end{align}
\endgroup
\end{proposition}

\begin{remark}
For an admissible pair $(\ga, \de)$, we have that $\ga\de$, $\thq(\ga\zm, \de\zm)$ and $\thq(\ga\zp, \de\zp)$ are all positive.
For the formulas above, $\sqrt{\thq(\ga\zm, \de\zm, \ga\zp, \de\zp)}$ and $\sqrt{\ga\de}$ are the positive square roots.
\end{remark}

\begin{proof}[Proof of Theorem \ref{thm:projection}]
We show that $\hatK(\eta)$ is a rank $1$ projection matrix.
For that, we prove three statements: (1) $\hatK(\eta)$ is Hermitian, (2) $\det \hatK(\eta) = 0$, and (3) $\text{tr}\ \hatK(\eta) = 1$.

The first statement is equivalent to the equalities
\begin{equation}\label{Khermitian}
\overline{\hatK_{+, +}(\eta)} \stackrel{?}{=} \hatK_{+, +}(\eta), \quad \overline{\hatK_{-, -}(\eta)} \stackrel{?}{=} \hatK_{-, -}(\eta), \quad \overline{\hatK_{+, -}(\eta)} \stackrel{?}{=} \hatK_{-, +}(\eta).
\end{equation}

If $(\ga, \de)$ belongs to the principal series, then $(\overline{\ga}, \overline{\de}) = (\de, \ga)$, whereas if $(\ga, \de)$ belongs to the complementary series, then $(\overline{\ga}, \overline{\de}) = (\ga, \de)$.
Together with the obvious $\overline{\thq(x)} = \thq(\overline{x})$, one can easily verify all three identities in \eqref{Khermitian} for pairs in both the principal and complementary series.

The second statement follows from \eqref{ppthm}--\eqref{mmthm} in a straightforward manner.

For the third statement, we need to show
\begin{multline*}
\zp\thq(\ga\zp, \de\zp, -e^{i\eta}\zm\sqrt{\ga\de/q}, -e^{-i\eta}\zm\sqrt{\ga\de/q})\\
- \zm\thq(\ga\zm, \de\zm, -e^{i\eta}\zp\sqrt{\ga\de/q}, -e^{-i\eta}\zp\sqrt{\ga\de/q})\\
\stackrel{?}{=} -\frac{\ga\de\zm\zp^2}{q} \cdot \thq(\zm/\zp, \ga\de\zm\zp, -e^{i\eta}\sqrt{q\ga\de}/\ga, -e^{i\eta}\sqrt{q\de\ga}/\de).
\end{multline*}
By the quasi-periodicity $(-\ga\de\zm\zp/q)\cdot\thq(\ga\de\zm\zp) = \thq(\ga\de\zm\zp/q)$, the equality above becomes
\begin{multline}\label{sumthetas_special}
\thq(\zm/\zp, \ga\de\zm\zp/q, -e^{i\eta}\sqrt{q\ga\de}/\ga, -e^{i\eta}\sqrt{q\ga\de}/\de)\\
- \thq(\ga\zp, \de\zp, -e^{i\eta}\zm\sqrt{\ga\de/q}, -e^{-i\eta}\zm\sqrt{\ga\de/q})\\
\stackrel{?}{=} -\frac{\zm}{\zp}\cdot\thq(\ga\zm, \de\zm, -e^{i\eta}\zp\sqrt{\ga\de/q}, -e^{-i\eta}\zp\sqrt{\ga\de/q}).
\end{multline}
This is a special case of Lemma \ref{lem:qelliptic}: upon setting $X = e^{i\eta}/\sqrt{q}$, $Y = -\zp\sqrt{\ga\de}/q$, $Z = -\zm\sqrt{\ga\de}/q$, and $W = -\sqrt{\ga\de}/\ga$, equation \eqref{sumthetas} gives \eqref{sumthetas_special}.
\end{proof}

\subsection{Factorization of two-sided sums}

The proof of Proposition \ref{prop:fourier} relies on two summation identities which we now present.
In this section, $p$ is a real number in $(0, 1)$.

\begin{lemma}\label{technical1}
Let $a\in\C$ be such that $p < |a| < p^{-1}$.
Then
\begin{equation}\label{Fpafactored}
\sum_{m=-\infty}^{+\infty}{\frac{a^m}{zp^m + z^{-1}p^{-m}}} = -z \cdot\frac{\theta_{p^2}(-apz^2)\theta_{p^2}'(1)}{\theta_{p^2}(-z^2)\theta_{p^2}(ap)}, \quad z\in\C^*.
\end{equation}
\end{lemma}
\begin{proof}
The equality can be deduced from Ramanujan's $_1\psi_1$--identity, as explained in \cite[Rem. 2.4]{B_2007}.
Yet another proof is given in \cite[Sec. 4]{BB}.
\end{proof}

\begin{lemma}\label{technical2}
\begin{equation}\label{Hp}
\sum_{m \in \Z \setminus \{0\}}{\frac{z^m}{p^{-m} - p^{m}}} = -pz \cdot\frac{\theta_{p^2}'(pz)}{\theta_{p^2}(pz)}, \quad p < |z| < p^{-1}.
\end{equation}
\end{lemma}
\begin{proof}
When $m \gg 0$, the $m$-th term in the sum is $\sim (zp)^m$;
when $m \ll 0$, the $m$-th term is $\sim -(z^{-1}p)^{-m}$.
Thus the sum is absolutely convergent and defines an analytic function on the domain $\{ z\in\C : |zp|, |z^{-1}p| < 1 \} = \{ z\in\C : p < |z| < p^{-1} \}$.

Let $H_p(z)$ and $F_{p, a}(z)$ denote the left-hand sides of \eqref{Hp} and \eqref{Fpafactored}, respectively.
Notice that
$$H_p(z) = -i \left.\left\{ F_{p, z}(y) - \frac{1}{y + y^{-1}} \right\}\right|_{y = i} = -i \left.\left\{ F_{p, z}(i y) - i \frac{y}{1 - y^2} \right\}\right|_{y = 1}.$$
From Lemma \ref{technical1}, we have
\begin{equation}\label{simplifyHp1}
H_p(z) = \left. -y \left\{  \frac{\theta_{p^2}(pzy^2)\theta_{p^2}'(1)}{\theta_{p^2}(y^2)\theta_{p^2}(pz)} + \frac{1}{1 - y^2}  \right\} \right|_{y = 1}
= \left. \left\{  \frac{\theta_{p^2}(pzy)(p^2; p^2)_{\infty}^2}{\theta_{p^2}(y)\theta_{p^2}(pz)} - \frac{1}{1 - y}  \right\} \right|_{y = 1}.
\end{equation}

Note that $f(y) := \theta_{p^2}(y)/(1 - y) = (p^2y, p^2/y; p^2)_{\infty}$ is analytic on $\C^*$.
Also let $g(y) := \theta_{p^2}(pzy)$, so \eqref{simplifyHp1} becomes
\begin{equation}\label{simplifyHp2}
H_p(z) = \left. \frac{1}{y - 1}\left\{  - \frac{g(y)(p^2; p^2)_{\infty}^2}{f(y)\theta_{p^2}(pz)} + 1 \right\} \right|_{y = 1}
= \frac{(f'(1)g(1) - g'(1)f(1))(p^2; p^2)_{\infty}^2}{f(1)^2\theta_{p^2}(pz)}.
\end{equation}
Take derivatives to $f(y) = f(1/y)$ to obtain $f'(y) = -y^{-2}f'(1/y)$, in particular $f'(1) = 0$.
Furthermore, $f(1) = (p^2; p^2)_{\infty}^2$ and $g'(1) = pz\theta_{p^2}'(pz)$.
Plugging these values into \eqref{simplifyHp2} yields \eqref{Hp}.
\end{proof}

\subsection{Fourier transform of the elliptic tail kernel: proof of Proposition \ref{prop:fourier}}\label{sec:prooffourier}

Recall the constant $C = C(\ga, \de)$ defined in \eqref{Cconst}.

\begin{lemma}\label{lem:fourier}
For any $\eta\in\R$, the Fourier transform $\hatK(\eta)$ is given by
\begin{align}
\hatK_{+, +}(\eta) &= C \left\{ \de\zp\frac{\thq'(\de\zp)}{\thq(\de\zp)} - \ga\zp\frac{\thq'(\ga\zp)}{\thq(\ga\zp)}\right.\label{hatpp}\\
&\left.\ +\ e^{i\eta}\frac{\sqrt{q\ga\de}}{\ga}\cdot\frac{\thq'(-e^{i\eta}\sqrt{q\ga\de}/\ga)}{\thq(-e^{i\eta}\sqrt{q\ga\de}/\ga)} - e^{i\eta}\frac{\sqrt{q\ga\de}}{\de}\cdot\frac{\thq'(-e^{i\eta}\sqrt{q\ga\de}/\de)}{\thq(-e^{i\eta}\sqrt{q\ga\de}/\de)}\right\},\nonumber\\
\hatK_{+, -}(\eta) &= \frac{C\sqrt{|\zp/\zm|}}{\sqrt{\thq(\ga\zp, \de\zp, \ga\zm, \de\zm)}} \frac{\thq'(1)}{\thq(\zp/\zm)}\label{hatpm}\\
&\times\left\{ \frac{\thq(\ga\zp, \de\zm, e^{i\eta}|\zp/\zm|\sqrt{q\ga\de}/\ga)}{\thq(-e^{i\eta}\sqrt{q\ga\de}/\ga)}
- \frac{\thq(\de\zp, \ga\zm, e^{i\eta}|\zp/\zm|\sqrt{q\ga\de}/\de)}{\thq(-e^{i\eta}\sqrt{q\ga\de}/\de)} \right\},\nonumber\\
\hatK_{-, +}(\eta) &= \frac{C\sqrt{|\zm/\zp|}}{\sqrt{\thq(\ga\zp, \de\zp, \ga\zm, \de\zm)}} \frac{\thq'(1)}{\thq(\zm/\zp)}\label{hatmp}\\
&\times\left\{ \frac{\thq(\ga\zp, \de\zm, e^{i\eta}|\zm/\zp|\sqrt{q\ga\de}/\de)}{\thq(-e^{i\eta}\sqrt{q\ga\de}/\de)} -
\frac{\thq(\de\zp, \ga\zm, e^{i\eta}|\zm/\zp|\sqrt{q\ga\de}/\ga)}{\thq(-e^{i\eta}\sqrt{q\ga\de}/\ga)} \right\},\nonumber\\
\hatK_{-, -}(\eta) &= C \left\{ \ga\zm\frac{\thq'(\ga\zm)}{\thq(\ga\zm)} - \de\zm\frac{\thq'(\de\zm)}{\thq(\de\zm)} \right.\label{hatmm}\\
&\left.\ +\ e^{i\eta}\frac{\sqrt{q\ga\de}}{\de}\cdot\frac{\thq'(-e^{i\eta}\sqrt{q\ga\de}/\de)}{\thq(-e^{i\eta}\sqrt{q\ga\de}/\de)}- e^{i\eta}\frac{\sqrt{q\ga\de}}{\ga}\cdot\frac{\thq'(-e^{i\eta}\sqrt{q\ga\de}/\ga)}{\thq(-e^{i\eta}\sqrt{q\de/\ga})} \right\}.\nonumber
\end{align}
\end{lemma}
\begin{proof}
For $x \in \Z \setminus \{0\}$, Lemma \ref{ppsimplifying} gives
\begin{align}
\K_{+, +}(x, 0) &= C (-1)^x \cdot \frac{(\sqrt{\ga\de}/\ga)^x - (\sqrt{\ga\de}/\de)^x}{q^{-x/2} - q^{x/2}},\label{pp}\\
\K_{-, -}(x, 0) &= C (-1)^x\cdot \frac{(\sqrt{\ga\de}/\de)^x - (\sqrt{\ga\de}/\ga)^x}{q^{-x/2} - q^{x/2}}.\label{mm}
\end{align}
From Lemma \ref{lem:zero} and \eqref{pp},
$$\hatK_{+, +}(\eta) = C\zp \left\{ \de\frac{\thq'(\de\zp)}{\thq(\de\zp)} - \ga\frac{\thq'(\ga\zp)}{\thq(\ga\zp)} \right\} + C \cdot \sum_{x\in\Z\setminus\{0\}}{ e^{i\eta x} \cdot \frac{(\sqrt{\ga\de}/\ga)^x - (\sqrt{\ga\de}/\de)^x}{q^{-x/2} - q^{x/2}} }.$$
Then \eqref{hatpp} follows from Lemma \ref{technical2} with $p = \sqrt{q}$.
Note that we need the following inequalities to apply Lemma \ref{technical2}:
$$q^{1/2} < \left|\frac{\sqrt{\ga\de}}{\ga}\right|,\ \left|\frac{\sqrt{\ga\de}}{\de}\right| < q^{-1/2}.$$
They are equivalent to $|q\ga/\de|, |q\de/\ga| < 1$, and these follow from the fact that $(\ga, \de)$ is an admissible pair.
Similarly, by using \eqref{mm}, we obtain \eqref{hatmm}.

On the other hand, for any $x\in\Z$, Lemma \ref{pmsimplifying} shows
\begin{align}
\displaystyle\K_{+, -}(x, 0) &= \frac{C (-1)^x}{\sqrt{\thq(\ga\zm, \ga\zp, \de\zm, \de\zp)}} \cdot \frac{\thq(\de\zp, \ga\zm) (\sqrt{\ga\de}/\de)^{x} - \thq(\de\zm, \ga\zp) (\sqrt{\ga\de}/\ga)^{x}}{|\zp/\zm|^{1/2} q^{x/2} + |\zm/\zp|^{1/2} q^{-x/2}},\label{pm}\\
\K_{-, +}(x, 0) &= \frac{C (-1)^x}{\sqrt{\thq(\ga\zm, \ga\zp, \de\zm, \de\zp)}} \cdot \frac{ \thq(\de\zp, \ga\zm) (\sqrt{\ga\de}/\ga)^x - \thq(\de\zm, \ga\zp) (\sqrt{\ga\de}/\de)^x }{|\zm/\zp|^{1/2} q^{x/2} + |\zp/\zm|^{1/2} q^{-x/2}}.\label{mp}
\end{align}
Then \eqref{hatpm} follows from \eqref{pm} and Lemma \ref{technical1} applied to $p = \sqrt{q}$ (the restrictions of Lemma \ref{technical1} are satisfied because $(\ga, \de)$ is an admissible pair).
Similarly, \eqref{mp} gives \eqref{hatmp}.
\end{proof}

A proof of Proposition \ref{prop:fourier} will be furnished by the verification that the formulas in \eqref{hatpp}--\eqref{hatmm} are equal to the formulas in \eqref{ppthm}--\eqref{mmthm}.

We'll need the equality
\begin{equation*}
\thq'(1) = \lim_{x \to 1}{(\thq(x) - \thq(1))/(x - 1)} = \lim_{x \to 1}{\thq(x)/(x - 1)} = \lim_{x \to 1}{-(qx, q/x; q)_{\infty}}
= -(q; q)_{\infty}^2.
\end{equation*}

Together with \eqref{Cconst}, it follows that \eqref{hatpm} $\stackrel{?}{=}$ \eqref{pmthm} is equivalent to
\begin{multline*}
\thq(\ga\zp, \de\zm, e^{i\eta}|\zp/\zm|\sqrt{q\ga\de}/\ga, -e^{i\eta}\sqrt{q\ga\de}/\de)
- \thq(\ga\zm, \de\zp, e^{i\eta}|\zp/\zm|\sqrt{q\ga\de}/\de, -e^{i\eta}\sqrt{q\ga\de}/\ga)\\
\stackrel{?}{=} \frac{q}{\de|\zm|}\cdot\thq(q\zp/\zm, q\ga/\de, -e^{i\eta}\zp\sqrt{\ga\de/q}, -e^{-i\eta}\zm\sqrt{\ga\de/q}).
\end{multline*}
This identity is a particular case of Lemma \ref{lem:qelliptic} when we specialize the variables as follows:
$$X = i e^{i\eta}\sqrt{|\zp/q\zm|}, \ Y = i \sqrt{\ga\de |\zm\zp|}/q, \ Z = -i \sqrt{|\zp/\zm|} \sqrt{\ga\de}/\de, \ W = -i\sqrt{|\zp/\zm|} \sqrt{\ga\de}/\ga.$$
One similarly shows \eqref{hatmp} = \eqref{mpthm}.

It remains to prove \eqref{hatpp} $\stackrel{?}{=}$ \eqref{ppthm} and \eqref{hatmm} $\stackrel{?}{=}$ \eqref{mmthm}.
Both proofs are similar to many proofs in the literature on identities between elliptic functions (see e.g. \cite{R}, \cite[Sec. 15]{Ba} and references therein), so let us only give a proof sketch of the former equality in the remainder of this section.

In both sides of the identity to prove, replace $e^{i\eta}, \ga, \de$, and $\sqrt{\ga\de}$ by $z, c^2, d^2$ and $cd$, respectively.
Denote the formula coming from \eqref{hatpp} by $f(z, c, d)$ and the one coming from \eqref{ppthm} by $g(z, c, d)$.
The advantage is that both $f$ and $g$ are now meromorphic functions on $(z, c, d)\in(\C^*)^3$.
We shall actually prove $f(z, c, d) = g(z, c, d)$, for all values $(z, c, d)\in (\C^*)^3$ for which both sides are defined and not only in the case that $(c^2, d^2)$ is an admissible pair.

From the quasi-periodicity of the theta function, one verifies that both $f$ and $g$, as functions of $z$, are (multiplicatively) periodic with period $q$, i.e., $f(qz, c, d) = f(z, c, d)$ and $g(qz, c, d) = g(z, c, d)$.
One can also check that both $f$ and $g$ have only simple poles at the points of the form $-\frac{c}{d}\cdot q^{m+\frac{1}{2}}$ or $-\frac{d}{c}\cdot q^{m+\frac{1}{2}}$, for some $m\in\Z$ (in the special case $c = d$, minor changes are needed in the argument).
Moreover, their residues at these poles are the same, for example
\begin{multline*}
\Res_{z = -\frac{c\sqrt{q}}{d}}{f(z, c, d)} = \Res_{z = -\frac{c\sqrt{q}}{d}}{g(z, c, d)}\\
= -\frac{c\sqrt{q}}{d} \cdot \frac{\thq(c^2\zm, d^2\zm, c^2\zp, d^2\zp)}{\zp\cdot\thq(\zm/\zp, c^2d^2\zm\zp)}
\cdot \frac{(d^2 - c^2)}{c^2d^2 (c^2/d^2, d^2/c^2, q, q; q)_{\infty}}.
\end{multline*}
It follows that the difference $f - g$ is analytic, as a function of $z$, on $\C^*$.
Since it is also periodic, then $f - g$ is bounded on $\C^*$.
Liouville's theorem implies that $f- g$ is independent of $z$, so now it suffices to prove $f(-1, c, d) \stackrel{?}{=} g(-1, c, d)$.

From the formula \eqref{Cconst} for $C(c^2, d^2)$, the equality $f(-1, c, d) \stackrel{?}{=} g(-1, c, d)$ is equivalent to
\begin{multline*}
d^2\zp\frac{\thq'(d^2\zp)}{\thq(d^2\zp)} - c^2\zp\frac{\thq'(c^2\zp)}{\thq(c^2\zp)}
+ \frac{\sqrt{q}c}{d}\frac{\thq'(\sqrt{q}c/d)}{\thq(\sqrt{q}c/d)}
- \frac{\sqrt{q}d}{c}\frac{\thq'(\sqrt{q}d/c)}{\thq(\sqrt{q}d/c)}\\
\stackrel{?}{=} \frac{q \cdot (q; q)_{\infty}^2}{\zp d^2} \frac{\thq(d^2/c^2) \thq(\zp cd/\sqrt{q})^2}{\thq(c^2\zp, d^2\zp)\thq(\sqrt{q}d/c)^2}.
\end{multline*}
From the definition of theta function, we deduce $\thq(z^2) = \thq(z, -z, \sqrt{q}z, -\sqrt{q}z)$ for $z\in\C^*$.
Then $\thq(d^2/c^2) = \thq(d/c, -d/c, \sqrt{q}d/c, -\sqrt{q}d/c)$, so the desired identity becomes
\begin{multline}\label{simple_elliptic}
 d^2\zp\frac{\thq'(d^2\zp)}{\thq(d^2\zp)} - c^2\zp\frac{\thq'(c^2\zp)}{\thq(c^2\zp)}
+ \frac{\sqrt{q}c}{d}\frac{\thq'(\sqrt{q}c/d)}{\thq(\sqrt{q}c/d)}
- \frac{\sqrt{q}d}{c}\frac{\thq'(\sqrt{q}d/c)}{\thq(\sqrt{q}d/c)} \\
\stackrel{?}{=} \frac{q \cdot (q; q)_{\infty}^2}{\zp d^2} \cdot \frac{\thq(d/c, -d/c, -\sqrt{q}d/c) \thq(\zp cd/\sqrt{q})^2}{\thq(c^2\zp, d^2\zp, \sqrt{q}d/c)}.
\end{multline}
Let $\ell(c, d)$ and $r(c, d)$ be the left hand side and right hand side of \eqref{simple_elliptic}, respectively.
As before, one verifies $\ell(qc, d) = \ell(c, d)$ and $r(qc, d) = r(c, d)$.
Moreover $\ell$ and $r$ are meromorphic functions of $c$, with only simple poles exactly at points of the form $\zp^{-1/2}q^{m/2}$, $-\zp^{-1/2}q^{m/2}$, $dq^{m+\frac{1}{2}}$, for some $m\in\Z$.
Also, the residues of both sides coincide at all the poles; for example, one verifies
\begin{equation*}
\Res_{c = d\sqrt{q}}{\ \ell(c, d)} = \Res_{c = d\sqrt{q}}{\ r(c, d)} = 2d\sqrt{q}.
\end{equation*}
Therefore, the difference $\ell(c, d) - r(c, d)$ is a constant independent of $c$, meaning that it will suffice to prove $\ell(c, d) = r(c, d)$ for some value of $c$.
Finally, verify $\ell(\sqrt{q}/(\zp d), d) = r(\sqrt{q}/(\zp d), d) = 0$.

\section{Degeneration to the matrix trigonometric kernel}\label{sec:continuous}

In this section and the next we often use the variable
$$r = r(q) := -\ln{q} > 0,$$
so that $r \to 0^+$ as $q \to 1^-$.
We also use the material in Appendix \ref{app:jacobi} on estimates for theta functions.

\subsection{The matrix trigonometric kernel}\label{sec:trigokernel}

Let $\bfc, \bfd \in \C$ be such that $\bfd = \overline{\bfc} \in \C\setminus \R$ or $m < \bfc, \bfd < m+1$, for some $m\in\Z$.
Let $\Y := \R \sqcup \R$ and, given $u\in\R$, denote the corresponding elements of $\Y$ by $u^{(1)}$ or $u^{(2)}$ (depending on the copy of the real line to which $u$ belongs).
The kernel $K_{q \to 1}^{\mathfrak{c}, \mathfrak{d}}$ on $\Y$ is defined by
\begin{multline*}
K_{q \to 1}^{\mathfrak{c}, \mathfrak{d}}(u^{(i)}, v^{(j)}) :=\\
\begin{cases}
    \displaystyle \frac{\sin(\pi\bfc)\sin(\pi\bfd)}{\pi \sin(\pi(\bfc - \bfd))} \cdot \frac{\sinh\left( \frac{(\bfc - \bfd)(u - v)}{2} \right) }{\sinh\left( \frac{u - v}{2} \right) }, \hfill\text{ if } (i, j) = (1, 1) \text{ or }(2, 2),\\[15pt]
    \displaystyle \frac{\sqrt{\sin(\pi\bfc)\sin(\pi\bfd)}}{\pi \sin(\pi(\bfc - \bfd))} \cdot \frac{\sin(\pi\bfc) \exp\left( \frac{(\bfc - \bfd)(u - v)}{2} \right) - \sin(\pi\bfd) \exp\left( \frac{(\bfc - \bfd)(v - u)}{2} \right)}{\exp\left( \frac{u - v}{2} \right) + \exp\left( \frac{v - u}{2} \right)},\hfill\text{ if } (i, j) = (1, 2),\\[15pt]
    \displaystyle \frac{\sqrt{\sin(\pi\bfc)\sin(\pi\bfd)}}{\pi \sin(\pi(\bfc - \bfd))} \cdot \frac{\sin(\pi\bfd) \exp\left( \frac{(\bfc - \bfd)(u - v)}{2} \right) - \sin(\pi\bfc) \exp\left( \frac{(\bfc - \bfd)(v - u)}{2} \right)}{\exp\left( \frac{u - v}{2} \right) + \exp\left( \frac{v - u}{2} \right)},\hfill\text{ if } (i, j) = (2, 1).
\end{cases}
\end{multline*}
When $(i, j) \in \{ (1, 1), (2, 2) \}$ and $u = v$, we define the kernel by continuity, namely
$$\left. \frac{\sinh\left( \frac{(\bfc - \bfd)(u - v)}{2} \right) }{\sinh\left( \frac{u - v}{2} \right) } \right|_{u = v} = \bfc - \bfd.$$
Note that the case $\bfc = \bfd \in (m, m+1)$, for some $m\in\Z$, is allowed, so one needs to correct the definition of $K_{q \to 1}^{\mathfrak{c}, \mathfrak{d}}$ because it is given by the indeterminate ratio $0/0$  in that case.
The correction is done by using L'H\^opital's rule; see \cite[Sec. 6]{BO_2005} for more details.

The kernel $K_{q \to 1}^{\mathfrak{c}, \mathfrak{d}}$ will be called the \emph{matrix trigonometric kernel}.
It has appeared previously in the literature, e.g. it is called the \emph{tail kernel} in \cite{BO_2005}; it is shown there that it arises as a limit of both the \emph{discrete hypergeometric kernel} and the \emph{Gamma kernel}\footnote{The parameters $z, z'$ in \cite{BO_2005} are exactly $\bfc, \bfd$ in our notation.}.

To obtain $K_{q \to 1}^{\mathfrak{c}, \mathfrak{d}}$ as a limit of the elliptic tail kernel, we have to modify $K^{\ga, \de}$.
Let $\nu: \L \rightarrow \{-1, +1\}$ be $\nu(\zm q^m) = \nu(\zp q^m) := (-1)^m$, and $\mathbf{K}^{\ga, \de}(x, y) := \nu(x)\nu(y)^{-1}K^{\ga, \de}(x, y)$. Then
$$
\hatk^{\ga, \de}(x, y) := \begin{cases}
	\displaystyle \delta_{x, y} - \mathbf{K}^{\ga, \de}(x, y), &\text{ if } x = \zp q^m, \ y = \zp q^n,\\
	\displaystyle -\mathbf{K}^{\ga, \de}(x, y), &\text{ if } x = \zm q^m, \ y = \zp q^n,\\
	\displaystyle \mathbf{K}^{\ga, \de}(x, y), &\text{ if } y = \zm q^n.
\end{cases}
$$

This construction has a simple probabilistic meaning.
Both kernels $\mathbf{K}^{\ga, \de}$ and $K^{\ga, \de}$ differ by a gauge transformation, so they define the same point process $\mathcal{P}$ on the two-sided $q$-lattice $\L = \zp q^{\Z} \sqcup \zm q^{\Z}$.
The kernel $\hatk^{\ga, \de}$ can also be shown to define a point process $\widehat{\mathcal{P}}$ on $\L$.
The processes $\mathcal{P}$ and $\widehat{\mathcal{P}}$ are related by the particle-hole involution on the positive part of the lattice $\zp q^{\Z}$:
$$\text{if }X \text{ is }\mathcal{P}\text{--distributed, then }X \triangle\zp q^{\Z} \text{ is }\widehat{\mathcal{P}}\text{--distributed.}$$
For a proof, see \cite[Appendix \S A.3]{BOO}.

\subsection{Limit to the matrix trigonometric kernel}

\begin{theorem}\label{thmlimitII}
Assume that $\bfc, \bfd \in \C^2$ satisfy either $\bfd = \overline{\bfc} \in \C \setminus \R$ or $m < \bfc, \bfd < m+1$, for some $m\in\Z$; also, $\bfz_-, \bfz_+\in\R$ are arbitrary.
Then $\hatk^{\ga, \de}$ degenerates to the matrix trigonometric kernel $K_{q \to 1}^{\mathfrak{c}, \mathfrak{d}}$ in the following limit regime:
\begin{equation}\label{eqn:regimeII}
\begin{gathered}
m = \lfloor (-\ln{q})^{-1}u \rfloor,\quad n = \lfloor (-\ln{q})^{-1}v \rfloor,\\
\zm = -q^{\bfz_-},\quad \zp = q^{\bfz_+},\quad \ga = q^{\bfc - \bfz_+},\quad \de = q^{\bfd - \bfz_+},\quad q \to 1^-.
\end{gathered}
\end{equation}

In other words, identify $\L$ with $\X=\Z\sqcup\Z$ via $\zp q^k \mapsto k^{(1)}$, $\zm q^{\ell} \mapsto {\ell}^{(2)}$, as before, so that $\hatk^{\ga, \de}$ becomes a function on $\X^2$. Then, in the regime \eqref{eqn:regimeII}, we have the pointwise limit
\begin{equation}\label{tail_limit}
 (-\ln{q})^{-1} \hatk^{\ga, \de} \left(\lfloor (-\ln{q})^{-1}u \rfloor^{(i)}, \lfloor (-\ln{q})^{-1}v \rfloor^{(j)} \right) \rightarrow K_{q \to 1}^{\mathfrak{c}, \mathfrak{d}}(u^{(i)}, v^{(j)}).
\end{equation}
\end{theorem}

\begin{remark}
In the limit regime \eqref{eqn:regimeII}, note that $(\ga, \de)$ is an admissible pair and $\ga, \de \to 1$, as $q \to 1^-$.
A similar result holds in the case $\ga = -q^{\bfc - \bfz_-}$, $\de = -q^{\bfd - \bfz_-}$; note that $\ga, \de \to -1$, as $q \to 1^-$, in that case.
\end{remark}

\begin{proof}[Proof of Theorem \ref{thmlimitII}]
We analyze $(-\ln{q})^{-1}\hatk^{\ga, \de}(m^{(i)}, n^{(j)})$, for $i, j\in\{1, 2\}$, $m := \lfloor (-\ln{q})^{-1} u \rfloor$, $n := \lfloor (-\ln{q})^{-1} v \rfloor$, using Lemmas \ref{ppsimplifying} and \ref{lem:zero}.
Throughout the proof, the notation $A \sim B$ means $\lim_{q \to 1^{-}}{A/B} = 1$.

\smallskip

\emph{\textbf{Step 1.}}
First, estimate the constant $C = C(\ga, \de)$. Write it as
\begin{equation}\label{Cthetas}
C = \frac{1}{\zp \ga}\times\frac{\thq(\ga\zm, \ga\zp, \de\zm, \de\zp)}{\thq(\zm/\zp, \ga\de\zm\zp, \de/\ga)}\times \frac{1}{(q; q)_{\infty}^2}.
\end{equation}
From Lemma \ref{thetacomplex}, we deduce
\begin{equation*}
\thq(\zm \ga) \sim e^{\frac{\pi^2}{6r}}, \quad
\thq(\zm \de) \sim e^{\frac{\pi^2}{6r}}.
\end{equation*}
On the other hand, from Lemma \ref{thetacomplex0}, we have
\begin{align*}
\thq(\zp \ga) &\sim -i e^{-\frac{\pi^2}{3r}}\left( e^{\pi\bfc i} - e^{-\pi\bfc i} \right)
= 2 e^{-\frac{\pi^2}{3r}} \sin(\pi\bfc),\\
\thq(\zp \de) &\sim -i e^{-\frac{\pi^2}{3r}}\left( e^{\pi\bfd i} - e^{-\pi\bfd i}\right)
= 2 e^{-\frac{\pi^2}{3r}} \sin(\pi\bfd).
\end{align*}
An estimate for $(q; q)_{\infty}$ is in Lemma \ref{thetapositive}.
From Lemma \ref{thetacomplex} again, we have
\begin{equation*}
\thq(\zm/\zp) \sim e^{\frac{\pi^2}{6r}},\quad \thq(\ga\de\zm\zp) \sim e^{\frac{\pi^2}{6r}},\quad \thq(\de/\ga) \sim -i e^{-\frac{\pi^2}{3r} + \pi i(\bfd - \bfc)}.
\end{equation*}
Finally, plugging all these estimates into \eqref{Cthetas}, we obtain
\begin{equation}\label{Cprincipal}
C \sim \frac{r}{\pi} \cdot \frac{\sin(\pi\bfc)\sin(\pi\bfd)}{\sin(\pi(\bfd - \bfc))}.
\end{equation}

\emph{\textbf{Step 2.}}
We now estimate $\hatk^{\ga, \de}(m^{(1)}, n^{(1)})$ and $\hatk^{\ga, \de}(m^{(2)}, n^{(2)})$, for $u\neq v$.

When $u \neq v$, it is not hard to verify that, in our desired limit regime, we have
\begin{equation}\label{++tail}
\frac{\frac{\ga^m \de^n}{(\sqrt{\ga\de})^{m+n}} - \frac{\ga^n \de^m}{(\sqrt{\ga\de})^{m+n}}}{q^{(m-n)/2} - q^{(n-m)/2}}
\sim \frac{\exp\left( \frac{(\bfc - \bfd)(v - u)}{2} \right) - \exp\left( \frac{(\bfc - \bfd)(u - v)}{2} \right)}{\exp\left( \frac{v - u}{2} \right) - \exp\left( \frac{u - v}{2} \right)}
= \frac{\sinh\left( \frac{(\bfc - \bfd)(u - v)}{2} \right)}{\sinh\left( \frac{u - v}{2} \right)}.
\end{equation}

Use \eqref{Cprincipal}, \eqref{++tail}, Lemma \ref{ppsimplifying} and the definition of $\hatk^{\ga, \de}$ to obtain
$$\hatk^{\ga, \de} \left(\lfloor (-\ln{q})^{-1}u \rfloor^{(1)}, \lfloor (-\ln{q})^{-1}v \rfloor^{(1)} \right)
\sim \frac{r\sin(\pi\bfc)\sin(\pi\bfd)}{\pi \sin(\pi(\bfc - \bfd))}\cdot
\frac{\sinh\left( \frac{(\bfc - \bfd)(u - v)}{2} \right)}{\sinh\left( \frac{u - v}{2} \right)}.$$
Multiplying the above estimate by $(-\ln{q})^{-1} = r^{-1}$ proves \eqref{tail_limit} for $i = j = 1$ and $u \neq v$.
Similarly, one can show \eqref{tail_limit} for $i = j = 2$ and $u \neq v$.

\emph{\textbf{Step 3.}}
Next we estimate $\hatk^{\ga, \de}(m^{(1)}, m^{(1)})$ and $\hatk^{\ga, \de}(m^{(2)}, m^{(2)})$, for any $u$.
We need to estimate
$$\zp \left\{ \de\frac{\thq'(\de\zp)}{\thq(\de\zp)} - \ga\frac{\thq'(\ga\zp)}{\thq(\ga\zp)} \right\} = f_{\de}'(1) - f_{\ga}'(1),$$
where $f_{\de}(x) := \ln{\thq(\de\zp x)}$, $f_{\ga}(x) := \ln{\thq(\ga\zp x)}$.

From Lemma \ref{thetacomplex0}, applied to $z = \zp\ga x$, we have
$$\thq(z) \sim -i e^{-\frac{\pi^2}{3r}-\frac{2\pi^2u^2}{r} + \frac{2\pi^2u}{r} + i\pi u}(1 - e^{-\frac{4\pi^2 u}{r}}), \text{ where }u = - \frac{\bfc r}{2\pi i} + \frac{\ln{x}}{2\pi i}.$$
This leads to
$$f'_{\ga}(1) = \left.\frac{d}{dx}\ln{\thq(z)}\right|_{x = 1} \sim -\bfc - \frac{\pi i}{r} + \frac{1}{2} - \frac{2\pi i}{r}\cdot \frac{e^{-\pi \bfc i}}{e^{\pi\bfc i} - e^{-\pi\bfc i}}.$$
Similarly,
$$f'_{\de}(1) \sim -\bfd + \frac{\pi i}{r} + \frac{1}{2} + \frac{2\pi i}{r}\cdot \frac{e^{\pi\bfd i}}{e^{-\pi\bfd i} - e^{\pi\bfd i}}.$$
Therefore
\begin{multline}\label{ppests}
\zp \left\{ \de\frac{\thq'(\de\zp)}{\thq(\de\zp)} - \ga\frac{\thq'(\ga\zp)}{\thq(\ga\zp)} \right\}
= f'_{\de}(1) - f'_{\ga}(1)\\
\sim \bfc - \bfd + \frac{2\pi i}{r} + \frac{2\pi i}{r} \left\{ \frac{e^{-\pi\bfc i}}{e^{\pi\bfc i} - e^{-\pi\bfc i}} + \frac{e^{\pi\bfd i}}{e^{-\pi\bfd i} - e^{\pi\bfd i}} \right\} = \bfc - \bfd - \frac{\pi}{r}\frac{\sin(\pi(\bfc - \bfd))}{\sin(\pi\bfc)\sin(\pi\bfd)}.
\end{multline}
Use \eqref{Cprincipal}, \eqref{ppests}, Lemma \ref{lem:zero} and the definition of $\hatk^{\ga, \de}$ to get
\begin{align*}
\hatk^{\ga, \de} \left(\lfloor (-\ln{q})^{-1}u \rfloor^{(1)}, \lfloor (-\ln{q})^{-1}u \rfloor^{(1)} \right) &\sim 1 + \frac{r\sin(\pi\bfc)\sin(\pi\bfd)}{\pi \sin(\pi(\bfc - \bfd))}\left( \bfc - \bfd - \frac{\pi\sin(\pi(\bfc - \bfd))}{r\sin(\pi\bfc)\sin(\pi\bfd)} \right)\\
&= \frac{r\sin(\pi\bfc)\sin(\pi\bfd)}{\pi \sin(\pi(\bfc - \bfd))} \cdot (\bfc - \bfd).
\end{align*}
Multiplying the above estimate by $(-\ln{q})^{-1} = r^{-1}$ proves \eqref{tail_limit} for $i = j = 1$ (and any $u$).
One similarly shows \eqref{tail_limit} for $i = j = 2$, but using Lemma \ref{thetacomplex} rather than Lemma \ref{thetacomplex0}.

\emph{\textbf{Step 4.}}
Finally, we estimate $\hatk^{\ga, \de}(m^{(1)}, n^{(2)})$ and $\hatk^{\ga, \de}(m^{(2)}, n^{(1)})$.
We need the asymptotics of
$$\frac{\thq(\zm\ga, \zp\de)}{\sqrt{\thq(\zm\ga, \zm\de, \zp\ga, \zp\de)}} \text{ and } \frac{\thq(\zm\de, \zp\ga)}{\sqrt{\thq(\zm\ga, \zm\de, \zp\ga, \zp\de)}}.$$
From the estimates of Step 1, we have
\begin{equation*}
\frac{\thq(\zm\de, \zp\ga)}{\sqrt{\thq(\zm\ga, \zm\de, \zp\ga, \zp\de)}} \sim \frac{\sin(\pi\bfc)}{\sqrt{\sin(\pi\bfc)\sin(\pi\bfd)}},\quad \frac{\thq(\zm\ga, \zp\de)}{\sqrt{\thq(\zm\ga, \zm\de, \zp\ga, \zp\de)}} \sim \frac{\sin(\pi\bfd)}{\sqrt{\sin(\pi\bfc)\sin(\pi\bfd)}}.
\end{equation*}
Finally, from the estimates above, together with \eqref{Cprincipal} and Lemma \ref{pmsimplifying}, the desired \eqref{tail_limit} is proved for $i = 1$, $j = 2$.
The case $i = 2$, $j = 1$ is handled similarly.
\end{proof}

\section{Degeneration to the discrete sine kernel}\label{sec:discrete}

\subsection{The discrete sine kernel}

Let $\phi \in (0, \pi)$ be arbitrary.
The discrete sine kernel (associated to $\phi$) on $\Z$ is
\begin{equation*}
K_{\textrm{sine}}^{\phi}(m, n) :=
\begin{cases}
	\displaystyle \frac{\sin(\phi(m - n))}{\pi(m - n)}, &\text{ if }m \neq n,\\
	\displaystyle \frac{\phi}{\pi}, &\text{ if }m = n.
\end{cases}
\end{equation*}
The sine kernel is translation--invariant.
It is the correlation kernel for the \emph{discrete sine process} on the lattice of integers, see e.g. \cite{BOO}, \cite{BG}.

\subsection{Limit to the discrete sine kernel}

Recall the gauge-transformed elliptic tail kernel $\tilK^{\ga, \de}(x, y)$ defined in \eqref{dfgauge}.

\begin{theorem}\label{thmlimitI}
Let $\varphi\in (0, \pi)$ and $s > 0$ be fixed.
The admissible pair $(\ga, \de)$ may vary, but always satisfying $\ga = \overline{\de} \in \C\setminus \R$.
In the limit regime
\begin{equation}\label{eqn:regimeI}
\begin{gathered}
\frac{\ln{\de} - \ln{\ga}}{2i} = \varphi,\\
|m|, |n| \to \infty,\ q \to 1^-,\text{ in such a way that } m - n \text{ is fixed, and } q^m, q^n \to s,
\end{gathered}
\end{equation}
one has the pointwise limits
\begin{align*}
\tilK^{\ga, \de}(\zp q^{m}, \zp q^{n}) &\to K^{\pi - \varphi}_{\text{sine}}(m, n),\\
\tilK^{\ga, \de}(\zm q^{m}, \zm q^{n}) &\to K^{\varphi}_{\text{sine}}(m, n).
\end{align*}
\end{theorem}

\begin{remark}
When $q \to 1^-$, the lattice $\L$ approximates any point in the real line.
Theorem \ref{thmlimitI} is saying that near any point $a\in\R\setminus\{0\}$, the point processes $M^{\ga, \de}$ associated to the kernels $\tilK^{\ga, \de}$ (or equivalently, associated to the elliptic tail kernels $K^{\ga, \de}$) weakly converge to a discrete sine process $\mathcal{P}_a$.
Moreover, the parameter of $\mathcal{P}_a$ depends only on the sign of $a$: if $a > 0$, the parameter is $\pi - \varphi$, whereas if $a < 0$, the parameter is $\varphi$.
Observe that the discrete sine process associated to $\pi - \varphi$ is obtained from the one with parameter $\varphi$ by the particle-hole involution on $\Z$.
\end{remark}

\begin{remark}
One can show that, for any pairwise distinct $a_1, \dots, a_k\in\R\setminus \{0\}$, the discrete sine processes $\mathcal{P}_{a_1}, \dots, \mathcal{P}_{a_k}$ (obtained as weak limits of the measures $M^{\ga, \de}$) are independent.
\end{remark}

\begin{proof}[Proof of Theorem \ref{thmlimitI}]
We analyze $\tilK^{\ga, \de}(\zp q^{m}, \zp q^{n})$ and $\tilK^{\ga, \de}(\zm q^{m}, \zm q^{n})$ using Lemmas \ref{ppsimplifying} and \ref{lem:zero}.
Throughout the proof, the notation $A \sim B$ means $\lim_{q \to 1^{-}}{A/B} = 1$.

\smallskip

\emph{\textbf{Step 1.}}
We first analyze $C(\ga, \de)$.

From Lemma \ref{thetapositive}, $(q; q)_{\infty}^{-2} \sim \frac{r e^{\pi^2/3r}}{2\pi}$.
From Lemma \ref{thetacomplex}, we obtain
\begin{gather*}
\thq(\ga\zm) \sim e^{\frac{\pi^2}{6r}-\frac{2\pi^2}{r}(c + z_-)^2 + i\pi(c + z_-)},\quad \thq(\ga\zp) \sim i e^{-\frac{\pi^2}{3r}-\frac{2\pi^2}{r}(c + z_+)^2 - \frac{2\pi^2}{r}(c+z_+) + i\pi(c+z_+)},\\
\thq(\de\zm) \sim e^{\frac{\pi^2}{6r}-\frac{2\pi^2}{r}(d + z_-)^2 + i\pi(d + z_-)},\quad \thq(\de\zp) \sim -i e^{-\frac{\pi^2}{3r}-\frac{2\pi^2}{r}(d + z_+)^2 + \frac{2\pi^2}{r}(d+z_+) + i\pi(d+z_+)},\\
\thq(\zm/\zp) \sim e^{\frac{\pi^2}{6r}-\frac{2\pi^2}{r}(z_- - z_+)^2 + i\pi(z_- - z_+)},\quad \thq(\ga\de\zm\zp) \sim e^{\frac{\pi^2}{6r}-\frac{2\pi^2}{r}(c + d + z_- + z_+)^2 + i\pi(c + d + z_- + z_+)},\\
\thq(\de/\ga) \sim -i e^{-\frac{\pi^2}{3r}-\frac{2\pi^2}{r}(d - c)^2 + \frac{2\pi^2}{r}(d - c) + i\pi(d - c)}.
\end{gather*}

Putting everything together in the formula \eqref{Cthetas} for $C = C(\ga, \de)$ yields
\begin{equation}\label{Cestimate}
C \sim \frac{ri}{2\pi}.
\end{equation}

\emph{\textbf{Step 2.}}
From the assumption \eqref{eqn:regimeI}, we have
\begin{multline}\label{kernel++}
(-1)^{m+n} \times \frac{\frac{\ga^m\de^n}{(\sqrt{\ga\de})^{m+n}} - \frac{\ga^n\de^m}{(\sqrt{\ga\de})^{m+n}}}{q^{(m - n)/2} - q^{(n - m)/2}}
= (-1)^{m+n} \cdot \frac{\exp(\varphi(n - m)i) - \exp(-\varphi(n - m)i)}{r(n - m)}\\
= (-1)^{n-m}\cdot\frac{2i \sin(\varphi(n - m))}{r(n - m)} = \frac{2i \sin((\varphi - \pi)(m - n))}{r(m - n)}.
\end{multline}
Combining this equality with the estimate \eqref{Cestimate} and Lemma \ref{ppsimplifying}, we have
$$\tilK^{\ga, \de}(\zp q^{m}, \zp q^{n}) \sim -\frac{\sin((\varphi - \pi)(m - n))}{\pi(m - n)} = \frac{\sin((\pi - \varphi)(m - n))}{\pi(m - n)},\text{ for }m \neq n.$$
Similarly,
$$\tilK^{\ga, \de}(\zm q^{m}, \zm q^{n}) \sim \frac{\sin(\varphi(m - n))}{\pi(m - n)},\text{ for }m \neq n.$$

\emph{\textbf{Step 3.}}
We still need to study the case $m = n$. Begin with the equality
$$\zp \left\{ \de\frac{\thq'(\de\zp)}{\thq(\de\zp)} - \ga\frac{\thq'(\ga\zp)}{\thq(\ga\zp)} \right\} = f_{\de}'(1) - f_{\ga}'(1),$$
where we denoted $f_{\ga}(w) := \ln{\thq(\ga\zp w)}$, $f_{\de}(w) := \ln{\thq(\de\zp w)}$.
From Lemma \ref{thetacomplex}, we have
$$\ln{\thq(x)} = \left( \frac{\pi^2}{6r} - \frac{2\pi^2u^2}{r} + i\pi u \right)\cdot (1 + o(1)), \text{ where } u = u(x) = \frac{\ln(-x)}{2\pi i}, \ |\Re{u}| < \frac{1}{2}.$$
It will be convenient to use the notation
\begin{equation}\label{notationfracs}
c := \frac{\ln{\ga}}{2\pi i},\quad d := \frac{\ln{\de}}{2\pi i},\quad z_- := \frac{\ln{|\zm|}}{2\pi i},\quad z_+ := \frac{\ln{\zp}}{2\pi i}.
\end{equation}
Use the previous estimate for $x = \ga\zp w$ and $u = \frac{\ln(-\ga\zp w)}{2\pi i} = c + z_+ + \frac{\ln{w}}{2\pi i} + \frac{1}{2}$; it yields
$$
f'_{\ga}(w) \sim -\frac{4\pi^2u}{r}\frac{du}{dw} + i\pi\frac{du}{dw}
= -\frac{4\pi^2}{r}\left( c + z_+ + \frac{\ln{w}}{2\pi i} + \frac{1}{2} \right) \left( \frac{1}{2\pi i w} \right) + i\pi \left( \frac{1}{2\pi i w} \right).
$$
Setting $w = 1$, we can simplify the formula to
$$f'_{\ga}(1) \sim \frac{2\pi i}{r}\left( c + z_+ + \frac{1}{2} \right) + \frac{1}{2}.$$
Similarly,
$$f'_{\de}(1) \sim \frac{2\pi i}{r}\left( d + z_+ - \frac{1}{2} \right) + \frac{1}{2},$$
and therefore
\begin{equation*}
\zp \left\{ \de\frac{\thq'(\de\zp)}{\thq(\de\zp)} - \ga\frac{\thq'(\ga\zp)}{\thq(\ga\zp)} \right\} \sim \frac{2\pi i}{r}( d - c - 1 ) = \frac{2\pi i}{r}\left( \frac{\ln{\de}}{2\pi i} - \frac{\ln{\ga}}{2\pi i} - 1 \right)
= \frac{2i (\varphi - \pi)}{r}.
\end{equation*}
Combining this estimate with \eqref{Cestimate} and Lemma \ref{lem:zero}, we have
$$\tilK^{\ga, \de}(\zp q^{m}, \zp q^{m}) \sim \frac{\pi - \varphi}{\pi},\text{ for any }m.$$
Similarly, we obtain
$$\tilK^{\ga, \de}(\zm q^{m}, \zm q^{m}) \sim \frac{\varphi}{\pi},\text{ for any }m.$$
\end{proof}

\begin{appendix}

\section{Jacobi's Imaginary Transformation}\label{app:jacobi}

The third Jacobi theta function (\cite{WW}) is the analytic function on $\C^*$ defined by \footnote{In contrast with the usual definition, we use the parameter $q^{1/2}$ and not $q$.}
$$\theta_3(z; q) := \sum_{n\in\Z}{z^n q^{n^2/2}} = (q, -\sqrt{q}z, -\sqrt{q}/z; q)_{\infty} = (q; q)_{\infty}\cdot\theta_q(-\sqrt{q}z).$$
The second equality is Jacobi's triple product identity (see \cite{GR}).

Let $r = r(q) := -\ln{q} > 0$. For $z\in\C^*$, let $u = \frac{\ln{z}}{2\pi i}$. Then \emph{Jacobi's Imaginary Transformation} is
\begin{equation}\label{JIT}
\theta_3(z; q) = \left( \frac{2 \pi}{r} \right)^{\frac{1}{2}} e^{-\frac{2\pi^2u^2}{r}} \cdot \theta_3\left( e^{\frac{4\pi^2u}{r}};\ e^{-\frac{4\pi^2}{r}} \right).
\end{equation}

As $q \to 1^-$, then $r \to 0^+$, and so $e^{-\frac{4\pi^2}{r}} \to 0^+$.
Thus, in principle, limits of the Jacobi theta function $\theta_3(z; q)$ when $q \to 1^-$ are related to the limits when $q \to 0^+$.
This relation allows us to prove the following estimates for $q$-Pochhammer symbols.

\begin{lemma}\label{thetapositive}
Set $q = e^{-r}$, then
$$(q; q)_{\infty} = \left( \frac{2\pi}{r} \right)^{\frac{1}{2}}e^{-\frac{\pi^2}{6r}}(1 + o(1)), \textrm{ as } r \to 0^+.$$
\end{lemma}
\begin{proof}
The definition of $\theta_3(z; q)$ gives $(1 + \sqrt{q}z)^{-1}\theta_3(z; q) = (q, -q\sqrt{q}z, -\sqrt{q}/z; q)_{\infty}$.
Then set $z = -\frac{e^{2\pi i \epsilon}}{\sqrt{q}}$; Jacobi's imaginary transformation yields
\begin{multline*}
(q, qe^{2\pi i \epsilon}, qe^{-2\pi i \epsilon}; q)_{\infty} = (1 - e^{2\pi i \epsilon})^{-1}
\left( \frac{2 \pi}{r} \right)^{\frac{1}{2}} e^{-\frac{2\pi^2}{r}(\frac{1}{2} + \epsilon - \frac{ri}{4\pi})^2} \cdot \theta_3\left( -e^{\frac{2\pi^2}{r} + \frac{4\pi^2\epsilon}{r}};\ e^{-\frac{4\pi^2}{r}} \right)\\
= \left( \frac{2 \pi}{r} \right)^{\frac{1}{2}} e^{-\frac{2\pi^2}{r}(\frac{1}{2} + \epsilon - \frac{ri}{4\pi})^2}
\cdot\frac{1 - e^{\frac{4\pi^2 \epsilon}{r}}}{1 - e^{2\pi i \epsilon}}
\times\prod_{n=1}^{\infty}{(1 - e^{-\frac{4\pi^2 n}{r}})(1 - e^{-\frac{4\pi^2 n}{r} + \frac{4\pi^2\epsilon}{r}})(1 - e^{-\frac{4\pi^2 n}{r} - \frac{4\pi^2\epsilon}{r}})}.
\end{multline*}
Take the limit $\epsilon \to 0^+$ to get
\begin{align*}
(q; q)_{\infty}^3 &= -i\left( \frac{2 \pi}{r} \right)^{\frac{3}{2}} e^{-\frac{2\pi^2}{r}(\frac{1}{2} - \frac{ri}{4\pi})^2} \times
\prod_{n=1}^{\infty}{(1 - e^{-\frac{4\pi^2 n}{r}})^3}\\
&= \left( \frac{2 \pi}{r} \right)^{\frac{3}{2}} e^{-\frac{\pi^2}{2r} + \frac{r}{8}} \times
\prod_{n=1}^{\infty}{(1 - e^{-\frac{4\pi^2 n}{r}})^3} = \left( \frac{2 \pi}{r} \right)^{\frac{3}{2}} e^{-\frac{\pi^2}{2r}} (1 + o(1)),\text{ as }r\rightarrow 0^+,
\end{align*}
from which the result follows.
\end{proof}

\begin{lemma}\label{thetacomplex0}
Set $q = e^{-r}$.
For $z \in \C \setminus \R_{\leq 0}$, set $u = u(z) := \frac{\ln{z}}{2\pi i}$.
Then
$$\theta_q(z) = i e^{- \frac{\pi^2}{3r}-\frac{2\pi^2 u^2}{r} - \frac{2\pi^2 u}{r} + i \pi u}(1 - e^{\frac{4\pi^2u}{r}})\cdot (1 + o(1)), \textrm{ as }r \to 0^+.$$
The estimate is uniform for $|\arg{z}| \leq \pi - \epsilon$, where $\epsilon > 0$ is arbitrary.
\end{lemma}
\begin{proof}
In terms of the Jacobi theta function, we have
\begin{equation}\label{thetasrelation}
\theta_q(z) = \frac{\theta_3(-z/\sqrt{q}; q)}{(q; q)_{\infty}}.
\end{equation}
First assume $-\pi < \arg{z} \leq 0$. Let $v = \frac{\ln{(-z/\sqrt{q})}}{2\pi i} = \frac{\ln{z} - (\ln{q})/2 + \pi i}{2\pi i} = u + \frac{r}{4\pi i} + \frac{1}{2}$, so that $0 \leq \Re u < \frac{1}{2}$.
Then, Jacobi's imaginary transformation and the definition of $\theta_3$ give
\begin{align}
\theta_3(-z/\sqrt{q}; q) &= \left( \frac{2\pi}{r} \right)^{\frac{1}{2}} e^{-\frac{2\pi^2}{r}(u + \frac{r}{4\pi i} + \frac{1}{2})^2}
\prod_{n = 1}^{\infty}{\left( 1 - e^{-\frac{4n\pi^2}{r}} \right)} \left( 1 - e^{\frac{4\pi^2}{r}(u - n)} \right) \left( 1 - e^{\frac{4\pi^2}{r}(u - n + 1)} \right)\nonumber\\
&= \left( \frac{2\pi}{r} \right)^{\frac{1}{2}} e^{-\frac{2\pi^2}{r}(u + \frac{r}{4\pi i} + \frac{1}{2})^2} (1 - e^{\frac{4\pi^2u}{r}}) \cdot (1 + o(1)), \text{ as }r\to 0^+.\label{estimate1}
\end{align}
From \eqref{thetasrelation}, \eqref{estimate1}, and Lemma \ref{thetapositive}, we obtain the desired result.
The case $0 \leq \arg{z} < \pi$ is analogous, and the statement about uniformity is evident.
\end{proof}

\begin{lemma}\label{thetacomplex}
Set $q = e^{-r}$.
For $z \in \C\setminus \R_{\geq 0}$, set $v = v(z) := \frac{\ln(-z)}{2\pi i}$ so that $-\frac{1}{2} < \Re v < \frac{1}{2}$.
Then
$$\theta_q(z) = e^{\frac{\pi^2}{6r}-\frac{2\pi^2v^2}{r} + i \pi v}\cdot (1 + o(1)), \textrm{ as }r \to 0^+.$$
The estimate is uniform for $|\arg z| \geq \epsilon$, where $\epsilon > 0$ is arbitrary.
\end{lemma}
\begin{proof}
The proof is similar to that of Lemma \ref{thetacomplex0}.
\end{proof}

\end{appendix}

\bigskip

Cesar Cuenca: Department of Mathematics, Caltech, Pasadena, CA, USA.

Email address: cesar.a.cuenk@gmail.com

\bigskip

Vadim Gorin: Department of Mathematics, MIT, Cambridge, MA, USA; Institute for Information Transmission Problems, Moscow, Russia.

Email address: vadicgor@gmail.com

\bigskip

Grigori Olshanski: Institute for Information Transmission Problems, Moscow, Russia;
Skolkovo Institute of Science and Technology, Moscow, Russia;
Higher School of Economics, Moscow, Russia.

Email address: olsh2007@gmail.com

\end{document}